\documentclass[sigconf,balance=false]{acmart}

\usepackage{algorithm}
\usepackage{algpseudocode}
\usepackage{booktabs}
\usepackage{multirow}
\usepackage{tabularx}
\usepackage{array}
\usepackage{graphicx}
\usepackage{placeins}
\usepackage{dblfloatfix}
\usepackage{afterpage}
\setcitestyle{numbers,sort&compress}
\newtheorem{proposition}{Proposition}
\algrenewcommand{\algorithmicrequire}{\textbf{Input:}}
\algrenewcommand{\algorithmicensure}{\textbf{Output:}}

\definecolor{PromptBorder}{HTML}{557A95}
\definecolor{PromptHeader}{HTML}{557A95}
\definecolor{PromptFill}{HTML}{F2F6F8}
\newcommand{\promptpanel}[2]{%
  \par\smallskip\noindent
  \fcolorbox{PromptBorder}{PromptFill}{%
    \begin{minipage}{\dimexpr\columnwidth-2\fboxsep-2\fboxrule\relax}
      \colorbox{PromptHeader}{%
        \parbox{\dimexpr\linewidth-2\fboxsep\relax}{\color{white}\bfseries #1}}%
      \vspace{3pt}\par
      {\scriptsize\ttfamily\raggedright #2\par}
    \end{minipage}}%
  \par\smallskip
}

\AtBeginDocument{%
  }

\setcopyright{none}
\acmConference[WWW '27]
{Proceedings of the ACM Web Conference 2027}
{May 10--14, 2027}
{Dublin, Ireland}

\begin{document}

\title[PolicyMediator]{PolicyMediator: Multi-Regional Policy Learning
through Mediation Agents with Evolving Skills in OPC Ecosystems}

\author{Jie Li}
\email{lijiestu@mail.nwpu.edu.cn}
\affiliation{%
  \institution{Northwestern Polytechnical University}
  \city{Xi'an}
  \country{China}
}
\author{Bin Guo}
\authornote{Corresponding authors.}
\affiliation{%
  \institution{Northwestern Polytechnical University}
  \city{Xi'an}
  \country{China}
}
\author{Georgina Curto}
\authornotemark[1]
\affiliation{%
  \institution{United Nations University Institute in Macau}
  \city{Macau}
  \country{China}
}
\author{Zikang Liu}
\affiliation{%
  \institution{Northwestern Polytechnical University}
  \city{Xi'an}
  \country{China}
}
\author{Mengqi Chen}
\affiliation{%
  \institution{Northwestern Polytechnical University}
  \city{Xi'an}
  \country{China}
}
\author{yu li}
\affiliation{%
  \institution{Northwestern Polytechnical University}
  \city{Xi'an}
  \country{China}
}

\begin{abstract}
One-person companies (OPCs) are emerging enterprises operated by individual entrepreneurs through digital platforms and LLM-based AI agents. These highly skilled entrepreneurs are characterized by being digital nomads, with limited reliance on local ecosystems. Governments in China use distinct regional support policies to attract mobile OPCs, turning local support into a cross-regional coordination problem. OPCs compare regions to maximize their own utility, but imperfect policy perceptions may cause them to reject otherwise suitable locations. Existing policy-learning methods optimize resources under fixed perceptions, while consensus-formation methods aggregate views without feeding cognitive changes back into policy learning. This disconnect produces mismatched location choices and inefficient regional policy adjustment. We introduce PolicyMediator, a cognition-constrained leader-multi-follower Stackelberg framework with bilateral mediation. The mediator explains verified policies to OPCs to build awareness and trust, and reports their needs and unresolved concerns to governments for policy learning. Its self-evolving skills retain, revise, or prune strategies using verification, allocation, capacity, and objective trajectories rather than mediation success alone. Across six-region simulations in China with five OPC types, five scenarios, and three OPC-simulator backbones, PolicyMediator ranks first in all 15 OPC-backbone--scenario combinations. Over the eight-baseline mean, it improves the aggregate government objective, combining economic performance, employment benefit, and policy-access equality, by $42.4\%$ and average OPC utility by $40.9\%$. 
\end{abstract}



\received{25 October 2026}

\maketitle
\section{Introduction}
\begin{figure}[t]
\centering
\includegraphics[width=1\columnwidth]{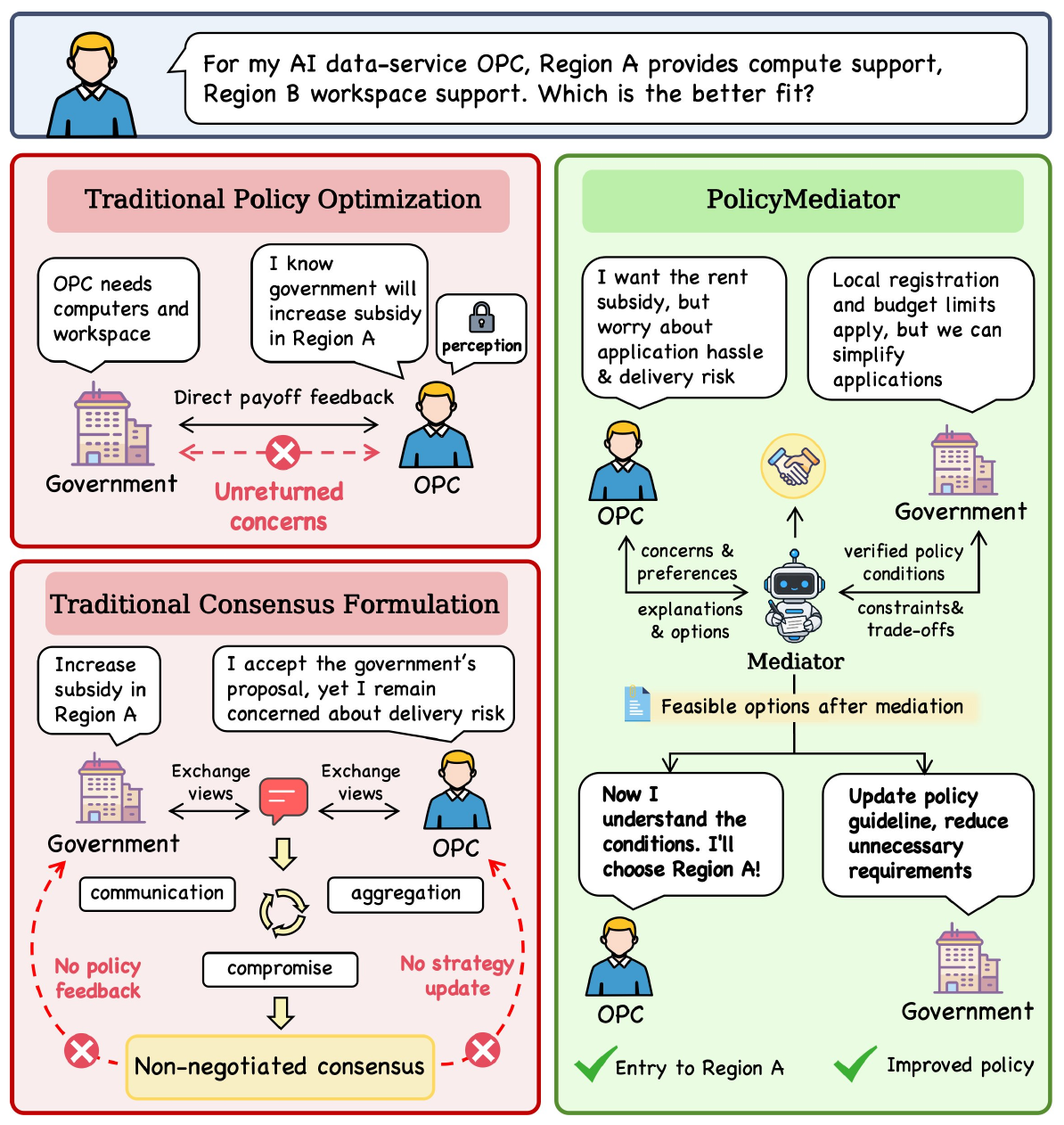}
\caption{Comparison of regional policy-coordination paradigms. Policy optimization excludes OPC concerns, while consensus formation lacks policy or cognition updates. PolicyMediator links verified mediation to OPC choices, policy revision, and skill evolution.}
\Description{Traditional policy optimization uses realized payoff feedback but leaves OPC concerns unreturned. Traditional consensus formation exchanges and aggregates views without policy feedback or strategy updates. PolicyMediator connects verified policy explanations and OPC concerns to feasible location choices and revised government policy.}
\label{fig:intro_comparison}
\end{figure}

LLM-based AI agents are reshaping firm organization by making individuals viable units of production\cite{kim2026generative}. One-person companies (OPCs) exemplify this organizational shift, as Web platforms, cloud services, and generative AI lower the costs of acquiring customers, accessing professional services, and completing complex tasks\cite{alyakoob2026market,guo2025skill}. Governmental organizations across China introduce support policies and entrepreneurial hubs to make OPC development a priority for entrepreneurship, employment, and AI-enabled industries. Platform design and digital work infrastructure shape microentrepreneurial participation\cite{alyakoob2026market,doggett2024digital}, while local labor-market conditions influence who benefits\cite{guo2025skill,jones2024gendered}. OPCs therefore operate in a Web-enabled entrepreneurial ecosystem, using Web-based tools and resources to work across regions, access policy information, and relocate with limited reliance on local infrastructure and teams.

OPCs’ cross-regional location choices make the effects of regional support policies interdependent. Regional governments in China publish distinct support packages to attract mobile OPCs, while each OPC compares regions and selects the location maximizing expected utility. Yet these choices rely on incomplete, uncertain, and trust-dependent policy perceptions. Policy disclosure and mediation shape what an OPC knows, expects to be delivered, and regards as relevant; entry and relocation alter regional conditions and policy value. Prior work shows that policy-information frictions affect firm participation in government support programs and the allocation of assistance~\cite{humphries2020information}, infrastructure shapes digital-market access~\cite{leong2024coordination,agarwal2024mobile}, strategic disclosure redirects participation~\cite{liu2024unintended}, and mediation identifies common ground without overriding independent decisions~\cite{wall2001mediation,susskind2018deliberative}. This creates an agent-mediated socio-economic system in which the challenge is to coordinate regional policy learning with OPCs' evolving policy perceptions and location choices to improve government outcomes and OPC utility.

Existing methods to solve these challenges broadly fall into two categories. Policy-simulation methods represent heterogeneous governments and firms with agent-based models\cite{farmer2009economy,li2024econagent}, then use evolutionary games or reinforcement learning to evaluate policy and allocation\cite{zheng2022aieconomist,huang2026policysim}. They often fix policy perceptions or infer needs only from choices, so implementation concerns may be mistaken for inadequate support. Consensus methods exchange and aggregate opinions\cite{tessler2024common}, while mediation diagnoses disagreement and establishes shared facts\cite{wall2001mediation,susskind2018deliberative,chen2026mediation}. They often reach local agreement through compromise, but rarely update decision-relevant cognition, feed consensus into policy learning, or adapt mediation strategies. Figure~\ref{fig:intro_comparison} illustrates these gaps: policy optimization reacts to realized payoffs but leaves perception-related concerns unreturned, while consensus formation aggregates exchanged views without revising policy or its mediation strategy. Cross-regional coordination instead requires adaptive consensus through bilateral mediation, namely verified common ground and feasible options that evolve with policy conditions, OPC cognition, and allocation outcomes.

Therefore, Addressing this coordination problem presents two key challenges:
\begin{itemize}
 \item How can policy learning distinguish cognitive barriers from policy mismatches across regions?
\item How can mediation build adaptive consensus while balancing OPC interests, economic performance, and welfare?
\end{itemize}

We model this interplay by combining heterogeneous firm profiles with cognitive pathways of policy acceptance\cite{aguilera2025population}. OPC cognition captures what firms know, expect to be delivered, and perceive as relevant, thereby shaping approximately rational entry and relocation. Regional governments are modeled as a single cooperative leader that optimizes each region's policy under observed policy-intensity bounds and congestion effects. Bilateral mediator uses explanations grounded in verified policy sources to improve OPC decisions, and translates OPC needs and unresolved concerns into structured feedback for government policy updates. Mediator agents retain evolving skills only when their allocation and objective effects remain beneficial. As shown in  Figure~\ref{fig:intro_comparison}, the mediator combines verified policy conditions with OPC concerns and preferences to produce feasible location options and policy revisions.

We propose PolicyMediator, which integrates multi-region Stackelberg policy learning with bilateral mediation. In each round, governments publish regional policy packages, and mediators help OPCs assess feasible options under verified conditions. OPCs then independently choose a region or non-entry based on their updated policy understanding and expected utility. The resulting allocations, government outcomes, and unresolved OPC concerns inform policy learning and mediation skill updates for the next round. This closes the loop between government decisions, mediated cognition, and OPC responses. The process continues until policies, allocations, and both government and OPC objectives stabilize. Consensus concerns verified facts and workable options; it permits disagreement and does not require unanimous policy acceptance.

Our contributions can be summarized as follows:
\begin{itemize}
\item To our knowledge, we are the first to formulate heterogeneous OPC entry and location choice as a cognition-constrained leader--multi-follower Stackelberg model that includes bounded responses, policy interdependence, congestion, and equal policy access.
\item We couple policy learning with bilateral mediation and outcome-aware skill evolution, allowing explanations to influence OPC choices and structured concerns to improve subsequent government policies.
\item Across six Chinese regions, five scenarios, and three OPC backbones, PolicyMediator ranks first on both metrics in all 15 settings, improving the government objective/OPC utility by $42.4\%/40.9\%$ over the eight-baseline mean and $6.3\%/15.3\%$ over the strongest baseline.
\end{itemize}
\vspace{-10pt}

\section{Related Work}
\textbf{Multi-Agent Policy Making.} Multi-agent policy making uses agent-based models to characterize heterogeneous actors\cite{farmer2009economy}, while intervention systems optimize evolving networks and urban plans\cite{papachristou2023dynamic,yong2026intelliplanner}. For strategic interactions, game-theoretic methods learn leader--follower policy responses\cite{bai2021sample,zhong2023stackelbergnash}, while reinforcement learning optimizes economic policies and taxation\cite{zheng2022aieconomist,mi2024taxai,zhang2024noregret}. LLM agents enrich simulated societies with memory, reflection, and planning\cite{park2023generative,li2024econagent}, while PolicySim anticipates behavioral responses to policy\cite{huang2026policysim}. With fixed policy perceptions, these methods cannot distinguish policy mismatch from cognition-driven rejection. PolicyMediator instead links mediated cognitive updates to allocation and subsequent policy learning.

\textbf{Consensus Formation and Mediation.} Consensus formation coordinates beliefs and interests through opinion-dynamics models of disagreement, coevolving beliefs, polarization reduction, and bi-level consensus under social influence\cite{bindel2011opinion,auletta2019consensus,musco2018polarization,tu2024bilevel}. Beyond opinion dynamics, mediation diagnoses disagreement and structures negotiation\cite{wall2001mediation}, while joint fact-finding establishes shared evidence before bargaining\cite{susskind2018deliberative}. Computational approaches support mediation by identifying common ground\cite{tessler2024common}, organizing dispute resolution\cite{chen2026mediation}, and coordinating LLM exchanges\cite{pitre2025consensagent}. Yet opinion aggregation and compromise can trade away individual interests without correcting the cognition that drives later decisions. PolicyMediator preserves heterogeneous choices and evaluates policy updates by their downstream government objective.

\begin{figure*}[t]
\centering
\includegraphics[width=0.96\textwidth]{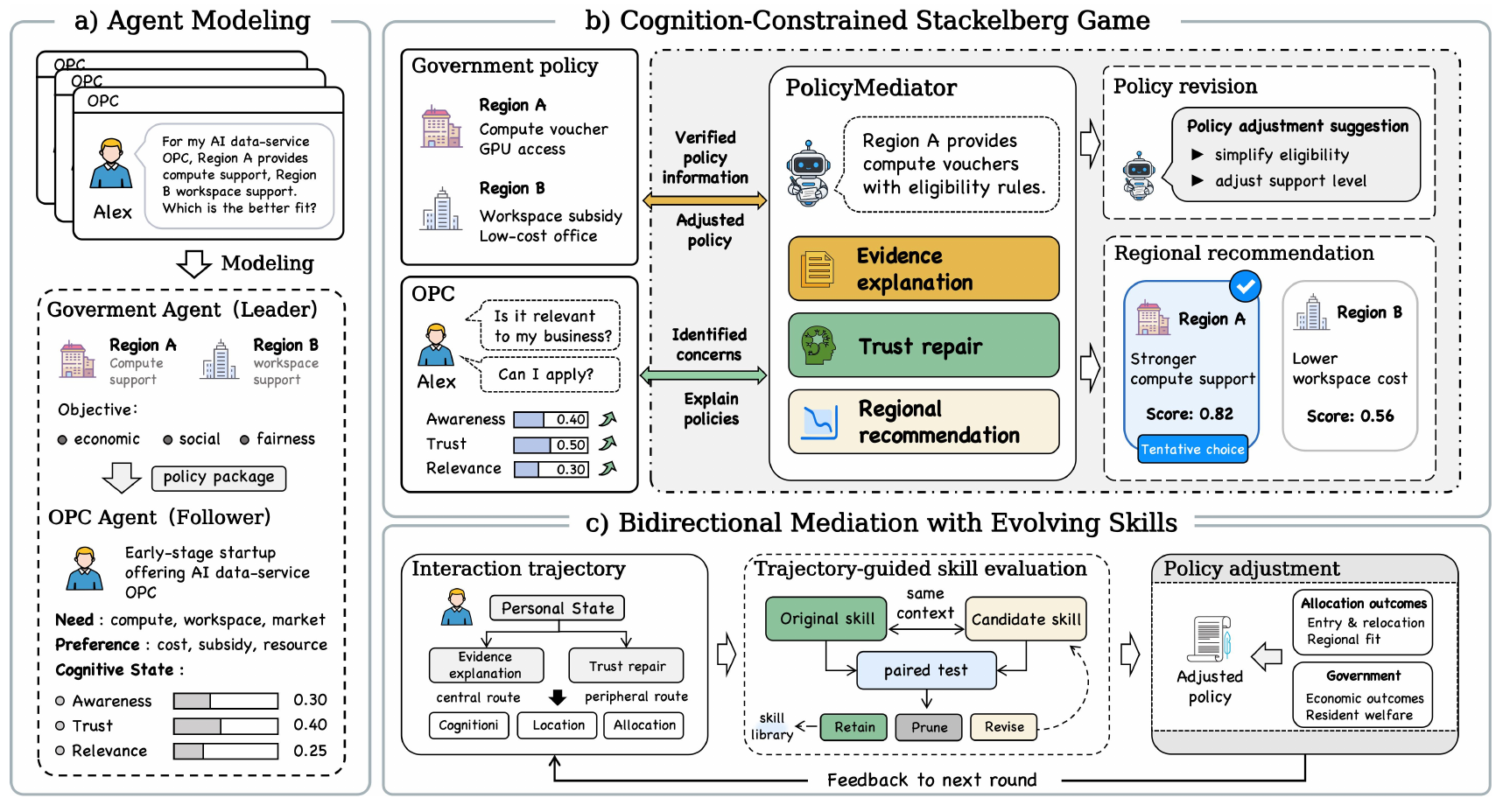}
\caption{PolicyMediator framework: (a) agent modeling; (b) cognition-constrained Stackelberg interaction with policy revision; and (c) trajectory-guided skill evaluation and next-round policy adjustment.}
\Description{Three-panel framework. The first panel models regional governments and heterogeneous OPCs. The second connects verified regional policies, mediation, cognition updates, location recommendations, and policy revision. The third evaluates original and candidate skills on matched trajectories, retains, revises, or prunes them, and feeds allocation and government outcomes into the next policy round.}
\label{fig:policymediator_framework}
\end{figure*}


\textbf{Self-Evolving Skills.} Self-evolving agents adapt reusable skills through accumulated experience. Foundational approaches formalize extended actions as options\cite{sutton1999options}, discover diverse behaviors through unsupervised learning\cite{eysenbach2019diayn}, and address retention and adaptation through continual learning\cite{wang2024continual}. Building on these foundations, language agents reflect on trajectories\cite{shinn2023reflexion}, refine outputs\cite{madaan2023selfrefine}, and learn from collaborative feedback\cite{zhang2024agentcf}. Recent skill libraries further support experience-driven revision and composition by organizing reusable abstractions and implementations\cite{yu2026polyskill,yang2026evods,wang2026skilllibrary,jiao2026agentic}. Task-success-driven evolution may retain persuasive skills with harmful downstream effects. PolicyMediator instead evaluates skills by factual validity and downstream system outcomes before retaining, revising, or pruning them.

\section{Problem Formulation}
\subsection{Framework Overview}
Figure~\ref{fig:policymediator_framework} presents the proposed PolicyMediator framework. Panel (a) models regional governments and heterogeneous OPCs. With regional policies centrally coordinated, governments across \(\mathcal R\) form a single cooperative leader while retaining distinct local policies. In panel (b), the leader publishes a feasible policy profile \(\mathbf p^k=(\mathbf p_r^k)_{r\in\mathcal R}\); mediators verify and explain its conditions, elicit concerns, and update cognition before OPCs choose a region or non-entry. Allocations and government outcomes inform policy revision. Panel (c) uses cognition, location, and allocation trajectories to compare original and candidate skills in matched contexts, retain, revise, or prune them, and feed evaluated outcomes into the next policy cycle.

Within each policy cycle, time is discretized into OPC decision steps $t\geq0$, with $k=\lfloor t/H\rfloor$ indexing $H$-step policy cycles.
Governments hold policies fixed within a cycle while OPCs update their choices. Let $x_{ir}^t=\mathbf1(x_i^t=r)$ and $x_i^t=(x_{ir}^t)_{r\in\mathcal R}$. For cycle $k$, $\mathcal T_k^\star$ is the longest suffix ending at $(k+1)H-1$ such that, at every transition, at most an $\epsilon$ fraction of OPCs switch regions. We write $M_k=|\mathcal T_k^\star|$ and use $\langle\cdot\rangle_k$ for averages over this stabilized window.

\subsection{Government agents.} Each government $G_r$ selects cash, finance, computing and data, space, talent, and market-access policies. Its package is $\mathbf p_r^k=(\mathbf y_r^k,\boldsymbol\theta_r^k)$, where $y_{rj}^k\in\{0,1\}$ indicates activation and $\theta_{rj}^k$ specifies intensity. For each category $j=1,\ldots,6$, the observed policy-derived bound $\widehat{\theta}_{rj}^{\mathrm{obs}}$ defines feasibility:
\begin{equation}
\mathbf p_r^k\in\mathcal F_r
\quad\Longleftrightarrow\quad
y_{rj}^k\in\{0,1\},\qquad
0\leq\theta_{rj}^k\leq
y_{rj}^k\widehat{\theta}_{rj}^{\mathrm{obs}}.
\label{eq:policy-feasible-set}
\end{equation}
At $k=0$, each region is initialized with its observed policy profile. Appendix~\ref{app:data-construction} details the construction, scaling, and initialization of these policy quantities.

\subsection{OPC agents and regional attributes.}

Each OPC chooses $a_i^t\in A=\mathcal R\cup\{0\}$, where $0$ denotes non-entry. Its type $l_i$ determines policy sensitivity through $\Gamma=(\gamma_{lj})\in[0,1]^{5\times6}$; Appendix~\ref{app:prompts} specifies the profile fields and the procedure for initializing cognition from verified regional evidence. Because acceptance depends on heterogeneous attributes and how policy information is processed~\cite{aguilera2025population,petty1986elm,gefen2003trust}, we represent policy-information-induced cognitive barriers by awareness $z_i^t$, trust in policy delivery $b_i^t$, and region-specific relevance $q_{ir}^t=\frac16\sum_{j=1}^6q_{irj}^t$, the mean of six policy-category perceptions for region $r$. Thus, $\Psi_{ir}^{\mathrm{opc},t}=(z_i^t,b_i^t,q_{ir}^t)$ and $\boldsymbol\Psi_i^t=(z_i^t,b_i^t,(q_{ir}^t)_{r\in\mathcal R})\in[0,1]^{2+|\mathcal R|}$.

Initial cognition $\boldsymbol\Psi_i^0$ is generated accordingly. At step $t$, the LLM analysis operator $\mathcal A_{\mathrm{LLM}}$ maps the profile, interaction history $\mathcal H_i^{\leq t}$, dialogue $\mathcal D_i^t$, and evidence $\mathcal S^k$ to the next cognitive update:
\begingroup
\setlength{\abovedisplayskip}{5pt plus 1pt minus 1pt}
\setlength{\abovedisplayshortskip}{4pt plus 1pt minus 1pt}
\setlength{\belowdisplayskip}{5pt plus 1pt minus 1pt}
\setlength{\belowdisplayshortskip}{5pt plus 1pt minus 1pt}
\begin{equation}
\begin{aligned}
\Delta\boldsymbol\Psi_i^{t+1}
&=\boldsymbol\delta_\Psi\odot\tanh\!\left(\mathcal A_{\mathrm{LLM}}
\left(\mathrm{profile}_i,\mathcal H_i^{\leq t},
\mathcal D_i^t,\mathcal S^k\right)\right),\\
\boldsymbol\Psi_i^{t+1}
&=\Pi_{[0,1]^{2+|\mathcal R|}}
\left(\boldsymbol\Psi_i^t
+\Delta\boldsymbol\Psi_i^{t+1}\right).
\end{aligned}
\label{eq:cognitive-transition}
\end{equation}
\endgroup
Here $\boldsymbol\delta_\Psi=(.15,.12,(.10)_{r\in\mathcal R})$ permits bidirectional updates, while $\mathcal S^k$ faithfully discloses verified negative as well as positive policy information. The second line accumulates these changes and projects every cognitive component back to $[0,1]$.

Because location value depends on earnings, industrial compatibility, administrative costs, and congestion, region $r$ is represented by baseline payoff $v_r$, type fit $m_{lr}$, application difficulty $d_r$, and capacity $K_r$; Appendix~\ref{app:data-construction} gives their construction. Type $l$ receives pre-policy return $m_{lr}v_r$. To represent OPC choice across regions, market-access support reduces cross-regional entry or relocation costs to $K_r^{{entry},k}=\max(d_r-\theta_{r6}^k,0)$.

\vspace{-0.35\baselineskip}
\section{Cognition-Constrained Stackelberg Game}
\subsection{Objective Function and Constraints}
\textbf{Objective Function of the OPCs.} Given $\mathbf p^k$, OPC $i$ evaluates region $r$ using its perceived policy benefit and the capacity penalty:
\begin{equation}
\begin{aligned}
B_{ir}^k&=\frac{1}{6}\sum_{j=1}^{6}\gamma_{l_i j}\theta_{rj}^k,\\
g_{ir}^t&=b_i^t\sqrt{z_i^tq_{ir}^t},\\
P_r^{\mathrm{cap},t-1}
&=\left[(n_r^{t-1}-K_r)/K_r\right]_+^2,\\
W_{ir}^{\mathrm{opc},t}
&=m_{l_i r}v_r+g_{ir}^tB_{ir}^k-K_r^{\mathrm{entry},k}
-P_r^{\mathrm{cap},t-1}.
\end{aligned}
\label{eq:opc-objective}
\end{equation}
Here, $W_{ir}^{\mathrm{opc},t}$ is OPC $i$'s perceived payoff, $g_{ir}^t\in[0,1]$ scales only its expected policy benefit, and $n_r^{t-1}=\sum_i x_{ir}^{t-1}$ is the latest observed occupancy. Initially, all OPCs begin with non-entry. Thus, OPCs respond to observed congestion. Reported OPC utility uses the same payoff with $g_{ir}^t$ replaced by policy fulfillment rate $\rho_r^k$ (default $1$), so cognition performance alone cannot raise evaluated outcomes. Following aggregate congestion modeling~\cite{rosenthal1973congestion}, the convex penalty is zero within $K_r$ and quadratic above it. We model OPC $i$'s boundedly rational choice using a logit response with temperature $\tau_i>0$:
\begin{equation}
\begin{aligned}
\rho_i^t(r)
&:=\Pr(a_i^t=r\mid s^t,\mathbf p^k)
=\frac{\exp\!\left(W_{ir}^{\mathrm{opc},t}/\tau_i\right)}
{\displaystyle\sum_{s\in A}
\exp\!\left(W_{is}^{\mathrm{opc},t}/\tau_i\right)},
\\
\widehat a_i^t
&\in\arg\max_{r\in A}\rho_i^t(r)
=\arg\max_{r\in A}W_{ir}^{\mathrm{opc},t}.
\end{aligned}
\label{eq:logit-response}
\end{equation}
Thus, $\widehat a_i^t$ is the highest-payoff city (or non-entry), while $a_i^t\sim\rho_i^t$ is realized. 
As $\tau_i\to 0$, the response approaches this maximum; positive temperature smooths response probabilities with respect to OPC payoffs, facilitating leader policy optimization.

\textbf{Objective Function of the Governments.} Following SDG~8 on productive employment and SDG~10 on equal opportunity and reduced inequalities of outcome~\cite{unitednations2015agenda}, we divide the government objective into economic performance and welfare; welfare comprises employment benefits and equal policy access across OPC types:
\begin{equation}
W^{G,k}
=\frac{1}{3|\mathcal R|}\sum_{r\in\mathcal R}
\left(L_r^{\mathrm E,k}+L_r^{\mathrm{emp},k}+L_r^{\mathrm{eq},k}\right),
\end{equation}
where the three components are placed on bounded, comparable scales and averaged equally. Reported $W^{G,k}$ values are scaled by $\zeta_W$. Let $\widehat R_{\max}=\max_{l,s}\widehat R_{ls}$ and $\bar u_{lr}^k=\frac{\rho_r^k}{6}\sum_j\gamma_{lj}\theta_{rj}^k$:
\begin{equation}
\begin{aligned}
L_r^{\mathrm E,k}
&=\frac{1}{NM_k}\sum_{t\in\mathcal T_k^\star}\sum_{i\in\mathcal N}
x_{ir}^tm_{l_i r}v_r-\widetilde C_r^{\mathrm{dir},k},\\
L_r^{\mathrm{emp},k}
&=\frac{\left\langle\sum_{i\in\mathcal N}x_{ir}^t
\widehat R_{l_i r}\right\rangle_k}{N\widehat R_{\max}+\epsilon_R},\\
L_r^{\mathrm{eq},k}
&=\frac{\left(\sum_{l=1}^{5}\bar u_{lr}^k\right)^3}
{25\sum_{l=1}^{5}(\bar u_{lr}^k)^2}.
\end{aligned}
\label{eq:welfare-components}
\end{equation}
The terms capture net economic performance, labor compensation, and equal policy access; $\widetilde C_r^{\mathrm{dir},k}$ is the direct fiscal expenditure of region $r$'s package, and the latter two terms define welfare. $L_{r}^{\mathrm{eq},k}$ measures equal policy access by multiplying Jain's fairness index by the mean policy benefit across OPC types. Appendix~\ref{app:data-construction} defines the proxies, normalization, and ranges.

\subsection{Response-Induced Markov Decision Process}
At cycle boundary $k$, the Markov state $s^k$ contains policies, allocation, cognition, regional conditions, and skill state, with $\mathbf o^k=\Omega(s^k)$, where $\Omega$ is the deterministic regional-observation map. Given feasible action $\mathbf u$, conditionally independent OPC choices induce $\rho(\mathbf a\mid s,\mathbf u)=\prod_i\rho_i(a_i\mid s,\mathbf u)$. Stackelberg here is generalized: a leader–response kernel in which followers respond boundedly rationally and payoffs
depend on aggregate state.

Following the leader--follower construction~\cite{zhong2023stackelbergnash}, integrating $W^G$ and the transition kernel over $\rho$ induces the return of the single cooperative leader
\begin{equation}
\begin{aligned}
J_G^\rho(\pi_G;s)
=\mathbb E\!\Bigg[&
\sum_{k=0}^{\infty}\delta_G^k W^{G,k}
\,\Bigm|\,s^0=s,\\[-2pt]
&\mathbf u^k=\pi_G(s^k),\ 
\mathbf a^k\sim\rho(\cdot\mid s^k,\mathbf u^k)
\Bigg].
\end{aligned}
\label{eq:leader-induced-return}
\end{equation}
The ideal response-conditioned policy is $\pi_G^\star\in\arg\max_{\pi_G\in\Pi_G}J_G^\rho(\pi_G;s)$.
Before developing an approximate solution, we establish conditions under which a stationary optimal leader policy exists.

\begin{proposition}[Stationary optimal leader policy under a logit response kernel]
\label{prop:stationary-leader-policy}
For the structured-state abstraction, if the state and feasible action spaces are compact, OPC payoffs are continuous, the response-averaged reward is bounded and upper semicontinuous, the response-averaged transition kernel is weakly continuous, $\tau_i>0$, and $\delta_G<1$, then a stationary optimal leader policy $\pi_G^\star$ exists.
\end{proposition}
The proof is provided in Appendix~\ref{app:proof-stationary-leader}. It establishes a theoretical target for the approximate policy learning developed next.

\subsection{Cooperative Policy Learning}
We use cooperative multi-actor zeroth-order policy-gradient learning to optimize regional policies, as mediated OPC responses make the simulator non-differentiable. Each actor $\pi_{\omega_r}$ maps local observation $o_r^k$ to adjustment $\Delta\boldsymbol\theta_r^k$, sharing the government objective as reward. Algorithm~\ref{alg:government-opc} shows paired parameter perturbations producing rollouts whose return differences update actors.

\textbf{Government Observations.} At boundary $t=kH$, $k\geq1$, government $G_r$ observes the previous policy and completed-cycle outcomes:
\begin{equation}
o_r^k=
\left(
\boldsymbol\theta_r^{k-1},
\bar{\mathbf x}_r^{k-1},
L_r^{\mathrm E,k-1},
L_r^{\mathrm{emp},k-1},
L_r^{\mathrm{eq},k-1},
\widetilde C_r^{\mathrm{dir},k-1},
W^{G,k-1}
\right),
\end{equation}
where $\bar{\mathbf x}_r^{k-1}=\langle(x_{ir}^t)_{i\in\mathcal N}\rangle_{k-1}$ and $\mathbf o^k=(o_r^k)_{r\in\mathcal R}$.

\textbf{Policy Generation.} Each regional government $G_r$ uses its deterministic actor $\pi_{\omega_r}$ to generate a bounded adjustment and projects the updated intensity onto the feasible set:
\begin{equation}
\begin{aligned}
\Delta\boldsymbol\theta_r^k
&=\operatorname{clip}\!\left(
\pi_{\omega_r}(o_r^k),
-\delta_\theta^{\max},\delta_\theta^{\max}
\right),\\
\boldsymbol\theta_r^k
&=\operatorname{Proj}_{\mathcal F_r^\theta}\!\left(
\boldsymbol\theta_r^{k-1}+\Delta\boldsymbol\theta_r^k
\right).
\end{aligned}
\label{eq:projected-policy-update}
\end{equation}
Clipping is componentwise, $0<\delta_\theta^{\max}\leq1$. Projection onto $\mathcal F_r^\theta$ sets $y_{rj}^k=\mathbf1(\theta_{rj}^k>0)$ and yields the package committed throughout cycle $k$; the pre-projection adjustment is the actor output, and projection belongs to the environment transition.

\textbf{Zeroth-Order Policy Learning.} Let $\boldsymbol\omega=(\omega_r)_{r\in\mathcal R}$ contain $D$ jointly perturbed actor parameters. At cycle $k$, $J_k(\boldsymbol\omega)$ is the expected rollout return from the saved simulator state, including its current allocation, cognition, and skill library. Following Gaussian zeroth-order optimization~\cite{nesterov2017gradientfree,ghadimi2013zerothorder}, define $J_{k,h}(\boldsymbol\omega)=\mathbb E_{\boldsymbol\zeta}[J_k(\boldsymbol\omega+h\boldsymbol\zeta)]$ for $\boldsymbol\zeta\sim\mathcal N(\mathbf0,I_D)$. Joint, two-sided rollouts start from the same saved state and use matched random seeds:
\begin{equation}
\begin{aligned}
&\widehat{\mathbf g}_k
=\frac{\widehat J_k(\boldsymbol\omega^k+h\boldsymbol\zeta^k)-\widehat J_k(\boldsymbol\omega^k-h\boldsymbol\zeta^k)}{2h}\,\boldsymbol\zeta^k,\\[-2pt]
&\omega_r^{k+1}=\omega_r^k+\alpha[\widehat{\mathbf g}_k]_r,
\end{aligned}
\label{eq:zeroth-order-policy-gradient}
\end{equation}
where each rollout recomputes mediation, skill evolution, OPC responses, congestion, and $W^G$ from the common checkpoint.
Since returns are stochastic and simulator states evolve, we analyze gradient estimation and stationarity under cross-cycle drift.
\begin{proposition}[Zeroth-order gradient estimation and stationarity]
\label{prop:zeroth-order-stationarity}
If paired returns are unbiased and lie in an interval of width $R$, then $\mathbb E_k[\widehat{\mathbf g}_k]=\nabla J_{k,h}(\boldsymbol\omega^k)$, $L_h\leq R/(2h^2)$, and the conditional estimator variance satisfies $\sigma_g^2\leq R^2D/(4h^2)$. Let $V_T=\sum_{k<T-1}\sup_{\boldsymbol\omega}|J_{k+1,h}(\boldsymbol\omega)-J_{k,h}(\boldsymbol\omega)|$. For $\alpha\leq1/L_h$,
\begin{equation}
\frac1T\sum_{k<T}\mathbb E\|\nabla J_{k,h}(\boldsymbol\omega^k)\|^2
\leq \frac{2(\Delta+V_T)}{\alpha T}+L_h\alpha\sigma_g^2,
\label{eq:zeroth-order-stationarity}
\end{equation}
where $\Delta=\sup_{k,\boldsymbol\omega}J_{k,h}(\boldsymbol\omega)-J_{0,h}(\boldsymbol\omega^0)$.
\end{proposition}
The result bounds the stationarity error of policy learning in terms of optimization progress, evolving simulator states, and rollout noise; Appendix~\ref{app:proof-zeroth-order} gives the proof. In practice, the realized $W^{G,k}$ guides the next policy update, and learning stops when policy, allocation, and $W^G$ remain stable for $M_{\mathrm{stab}}$ consecutive cycles.
\begin{algorithm}[t]
\caption{Government--OPC Dynamic Interaction}
\label{alg:government-opc}
\begin{algorithmic}[1]
\Require $\mathbf p^0$, OPC profiles/states, $\mathcal L^0$, $\{\pi_{\omega_r}\}_r$, $\alpha$, $h$, $T_G$, $H$, $M_{\mathrm{stab}}$
\Ensure Final policy, allocation, and $W^G$
\State Xavier-initialize hidden layers and zero-initialize actor output layers
\For{$k=0,\ldots,T_G-1$}
    \State Save checkpoint; draw $\boldsymbol\zeta^k\sim\mathcal N(\mathbf0,I_D)$
    \State If $k>0$, run matched $\pm h\boldsymbol\zeta^k$ rollouts from checkpoint; update actors by Eq.~\eqref{eq:zeroth-order-policy-gradient}
    \State Restore checkpoint; generate and publish feasible policy $\mathbf p^k$
    \For{$t=kH,\ldots,(k+1)H-1$}
        \State Apply verified skills; update cognition by Eq.~\eqref{eq:cognitive-transition}; compute payoffs
        \State Sample $a_i^t\sim\rho_i^t(\cdot\mid s^t,\mathbf p^k)$ and update states
    \EndFor
    \State Find $\mathcal T_k^\star$; evaluate $W^{G,k}$; construct $\tilde{\mathbf o}^{k+1}$ and evolve $\mathcal L^{k+1}$
    \State If policy, allocation, and $W^G$ stabilize for $M_{\mathrm{stab}}$ cycles, \textbf{break}
\EndFor
\State \Return $\mathbf p^k$, $\mathbf x^{(k+1)H}$, $W^{G,k}$
\end{algorithmic}
\end{algorithm}

\section{Bidirectional Mediation with Evolving Skills}
Our preliminary study (Section~\ref{sec:preliminary-study}) found that regional suitability and heterogeneous policy perceptions jointly shape OPC entry, leading us to design a mediator that addresses cognitive barriers through context-adaptive skills. This section details how these skills are structured under evidence constraints, deployed in bilateral interactions, and refined through outcome-based evaluation.

\subsection{Skill Structure}
A skill maps an OPC profile, concern, and verified evidence to an explanation, feasible recommendation, and structured government feedback, while recording applicability and outcomes. We manually initialize $\mathcal L^0$ with 14 skills distilled from OPC community experience; new candidates remain temporary until repeated evaluation supports retention. Appendix~\ref{app:prompts} gives the schema.

Each skill combines evidence constraints with a mediation strategy. Following the Context--Mechanism--Outcome framework~\cite{dalkin2015mechanism}, it conditions messages on capacity and industrial fit and tracks cognition, allocation, and $W^G$. Let $\sigma_{ri}^t$ be the message to OPC $i$ and $\mathcal S_{ri}^k$ its cited regional evidence. It must satisfy
\begin{equation}
\operatorname{Verify}(\sigma_{ri}^t,\mathcal S_{ri}^k)=1,
\qquad \mathcal S_{ri}^k\subseteq\mathcal S_r^k.
\end{equation}
Thus, $\operatorname{Verify}=1$ only when every factual claim is traceable to the current verified evidence set. Unsupported claims or omitted material conditions trigger revision, and capacity and local evidence are rechecked whenever a skill is invoked or reused.

While evidence constraints ensure factual reliability, they cannot address differences in OPC information processing. We therefore adopt the Elaboration Likelihood Model~\cite{petty1986elm} to guide central-route reasoning, peripheral credibility cues, and adaptive route switching. Accordingly, each skill comprises applicability conditions, an evidence boundary, a mediation route, feasible recommendation branches, and an outcome record.

\afterpage{
\begin{table*}[!t]
\centering
\caption{$W^G$ and OPC utility across scenarios and OPC-simulator backbones.}
\label{tab:joint_welfare_comparison}
\setlength{\tabcolsep}{3.2pt}\renewcommand{\arraystretch}{1.00}\small
\resizebox{\textwidth}{!}{\begin{tabular}{c|l|cccccccccc|cc}\hline
\multirow{2}{*}{\textbf{\shortstack{OPC\\Backbone}}}&\multirow{2}{*}{\textbf{Baseline}}&\multicolumn{2}{c}{\textbf{S1}}&\multicolumn{2}{c}{\textbf{S2}}&\multicolumn{2}{c}{\textbf{S3}}&\multicolumn{2}{c}{\textbf{S4}}&\multicolumn{2}{c}{\textbf{S5}}&\multicolumn{2}{c}{\textbf{Average}}\\
\cline{3-14}& & $\mathbf{W}_G$ & $\mathbf{W}_{\mathrm{OPC}}$ & $\mathbf{W}_G$ & $\mathbf{W}_{\mathrm{OPC}}$ & $\mathbf{W}_G$ & $\mathbf{W}_{\mathrm{OPC}}$ & $\mathbf{W}_G$ & $\mathbf{W}_{\mathrm{OPC}}$ & $\mathbf{W}_G$ & $\mathbf{W}_{\mathrm{OPC}}$ & $\mathbf{W}_G$ & $\mathbf{W}_{\mathrm{OPC}}$\\\hline
\multirow{9}{*}{\textbf{\shortstack{Qwen3.5-9B\\(Mar. 2026)}}}& No Policy&1.08&0.320&0.00&0.000&1.08&0.320&3.81&0.449&0.82&0.311&1.36&0.280\\
& Static Welfare&9.07&0.604&8.16&0.509&6.23&0.428&9.11&0.696&9.62&0.646&8.44&0.577\\
& Contextual Bandit&8.96&0.590&7.62&0.448&5.70&0.404&8.51&0.653&9.41&0.622&8.04&0.543\\
& Carbon-PPO&9.97&0.669&9.81&0.557&6.57&0.451&9.92&0.739&10.60&0.710&9.37&0.625\\
& Nash Welfare&8.25&0.582&7.16&0.518&5.88&0.435&8.67&0.682&8.32&0.647&7.66&0.573\\
& Belief-Robust&8.76&0.583&6.61&0.526&4.78&0.379&7.67&0.569&7.50&0.593&7.06&0.530\\
& Bi-level Consensus&8.50&0.544&6.55&0.417&5.54&0.391&8.63&0.641&8.50&0.576&7.54&0.514\\
& CONSENSAGENT&6.61&0.449&3.41&0.181&4.28&0.388&7.80&0.588&7.19&0.508&5.86&0.423\\
& \textbf{PolicyMediator}&\textbf{10.58}&\textbf{0.758}&\textbf{10.42}&\textbf{0.670}&\textbf{7.11}&\textbf{0.564}&\textbf{10.56}&\textbf{0.820}&\textbf{11.17}&\textbf{0.787}&\textbf{9.97}&\textbf{0.720}\\\hline
\multirow{9}{*}{\textbf{\shortstack{Ministral 3-8B\\(Dec. 2025)}}}& No Policy&1.08&0.320&0.00&0.000&1.08&0.320&3.81&0.449&0.82&0.311&1.36&0.280\\
& Static Welfare&9.06&0.602&8.11&0.509&6.18&0.421&9.10&0.694&9.61&0.643&8.41&0.574\\
& Contextual Bandit&8.96&0.588&7.61&0.445&5.64&0.402&8.50&0.651&9.39&0.619&8.02&0.541\\
& Carbon-PPO&9.96&0.667&9.79&0.554&6.56&0.449&9.90&0.737&10.58&0.706&9.36&0.623\\
& Nash Welfare&8.24&0.581&7.11&0.520&5.88&0.434&8.66&0.680&8.24&0.643&7.63&0.572\\
& Belief-Robust&8.75&0.581&6.85&0.420&4.70&0.376&7.66&0.568&8.24&0.623&7.24&0.514\\
& Bi-level Consensus&8.49&0.543&6.55&0.416&5.54&0.390&8.62&0.640&8.50&0.574&7.54&0.513\\
& CONSENSAGENT&7.83&0.517&5.66&0.379&5.23&0.382&8.43&0.626&8.00&0.558&7.03&0.492\\
& \textbf{PolicyMediator}&\textbf{10.54}&\textbf{0.760}&\textbf{10.37}&\textbf{0.670}&\textbf{7.10}&\textbf{0.563}&\textbf{10.55}&\textbf{0.817}&\textbf{11.15}&\textbf{0.784}&\textbf{9.94}&\textbf{0.719}\\\hline
\multirow{9}{*}{\textbf{\shortstack{Gemma 4-12B\\(Jun. 2026)}}}& No Policy&1.08&0.320&0.00&0.000&1.08&0.320&3.81&0.449&0.82&0.311&1.36&0.280\\
& Static Welfare&9.07&0.603&8.16&0.509&6.19&0.422&9.10&0.694&9.62&0.646&8.43&0.575\\
& Contextual Bandit&8.96&0.589&7.62&0.448&5.64&0.403&8.50&0.651&9.41&0.622&8.03&0.543\\
& Carbon-PPO&9.96&0.669&9.81&0.557&6.56&0.450&9.90&0.737&10.56&0.706&9.36&0.624\\
& Nash Welfare&8.25&0.582&7.16&0.518&5.88&0.435&8.66&0.680&8.32&0.648&7.65&0.573\\
& Belief-Robust&8.76&0.583&6.85&0.422&4.97&0.411&7.66&0.568&7.96&0.616&7.24&0.520\\
& Bi-level Consensus&8.50&0.544&6.55&0.417&5.54&0.390&8.63&0.641&8.50&0.575&7.54&0.513\\
& CONSENSAGENT&6.91&0.467&5.10&0.345&4.57&0.390&7.97&0.603&6.74&0.482&6.26&0.457\\
& \textbf{PolicyMediator}&\textbf{10.54}&\textbf{0.759}&\textbf{10.42}&\textbf{0.670}&\textbf{7.11}&\textbf{0.564}&\textbf{10.55}&\textbf{0.818}&\textbf{11.18}&\textbf{0.787}&\textbf{9.96}&\textbf{0.720}\\\hline
\end{tabular}}
\vspace{8pt}
\end{table*}
\begin{table}[!t]
\centering
\caption{Ablation results for $W^G$ (mean $\pm$ 95\% CI).}
\label{tab:ablation}
\renewcommand{\arraystretch}{1.15}
\small
\setlength{\tabcolsep}{2.5pt}
\begin{tabular*}{\columnwidth}{@{\extracolsep{\fill}}lccccc@{}}
\toprule
\textbf{Method} & \textbf{S1} & \textbf{S2} & \textbf{S3} & \textbf{S4} & \textbf{S5} \\
\midrule
\textbf{Full} & $10.58_{\scriptscriptstyle \pm 0.12}$ & $10.42_{\scriptscriptstyle \pm 0.13}$ & $7.11_{\scriptscriptstyle \pm 0.12}$ & $10.56_{\scriptscriptstyle \pm 0.11}$ & $11.17_{\scriptscriptstyle \pm 0.13}$ \\
\textbf{w/o Mediator} & $8.65_{\scriptscriptstyle \pm 0.11}$ & $8.29_{\scriptscriptstyle \pm 0.10}$ & $5.34_{\scriptscriptstyle \pm 0.07}$ & $8.23_{\scriptscriptstyle \pm 0.18}$ & $8.37_{\scriptscriptstyle \pm 0.12}$ \\
\textbf{w/o Cog. Update} & $8.88_{\scriptscriptstyle \pm 0.14}$ & $8.33_{\scriptscriptstyle \pm 0.09}$ & $5.49_{\scriptscriptstyle \pm 0.09}$ & $9.07_{\scriptscriptstyle \pm 0.15}$ & $8.75_{\scriptscriptstyle \pm 0.08}$ \\
\textbf{w/o Gov. Feedback} & $8.22_{\scriptscriptstyle \pm 0.14}$ & $7.52_{\scriptscriptstyle \pm 0.12}$ & $5.69_{\scriptscriptstyle \pm 0.13}$ & $8.56_{\scriptscriptstyle \pm 0.12}$ & $8.30_{\scriptscriptstyle \pm 0.10}$ \\
\textbf{w/o Skill Evol.} & $9.37_{\scriptscriptstyle \pm 0.14}$ & $9.03_{\scriptscriptstyle \pm 0.13}$ & $5.99_{\scriptscriptstyle \pm 0.12}$ & $9.70_{\scriptscriptstyle \pm 0.13}$ & $9.53_{\scriptscriptstyle \pm 0.10}$ \\
\textbf{w/o Theory Guide} & $9.73_{\scriptscriptstyle \pm 0.13}$ & $9.34_{\scriptscriptstyle \pm 0.13}$ & $6.25_{\scriptscriptstyle \pm 0.12}$ & $9.88_{\scriptscriptstyle \pm 0.11}$ & $9.91_{\scriptscriptstyle \pm 0.11}$ \\
\bottomrule
\end{tabular*}
\end{table}
}

\subsection{Mediation-Embedded Game Pipeline}
\textbf{OPC-Facing Dialogue.} Because cognition-constrained entry requires credible evidence rather than generic persuasion, each mediator retrieves a skill using the policy context, implementation outcomes, OPC profile, cognitive state, and regional evidence after $\mathbf p^k$ is published. At step $t$, it composes
\begin{equation}
\sigma_{ri}^t
=\operatorname{Compose}
\left(
\ell_{ri}^t,
\mathcal S_{ri}^k,
o_r^k,
\nu_{ri}^t
\right),
\qquad
\ell_{ri}^t\in
\mathcal L^k\cup\widetilde{\mathcal L}^k.
\end{equation}
where $\operatorname{Compose}$ is the evidence-constrained message constructor and $\nu_{ri}^t\in\mathcal R\cup\{0\}$ denotes the recommended region, an alternative, or non-entry. Only verified messages enter $\mathcal D_i^t$ and update cognition through Eq.~\eqref{eq:cognitive-transition}; the OPC then samples its own action using Eq.~\eqref{eq:logit-response}.

\textbf{Government-Facing Feedback.} Because unresolved concerns can indicate policy mismatch rather than misunderstanding, the mediator summarizes unmet needs, eligibility difficulties, delivery concerns, and observed responses after cycle $k$. A government-facing skill encodes this evidence as
\begin{equation}
\Delta\boldsymbol\theta_{r,\mathrm{LLM}}^{k+1}
=\operatorname{clip}\!\left(
\operatorname{Encoder}_{\mathrm{LLM}}\!\left(
o_r^k,
\mathcal D_r^k,
\Delta x^k,
W^{G,k}
\right),
\pm\delta_{\mathrm{LLM}}^{\max}
\right),
\end{equation}
where $\operatorname{Encoder}_{\mathrm{LLM}}$ is the government-feedback encoder, $0<\delta_{\mathrm{LLM}}^{\max}\leq1$, $\mathcal D_r^k$ collects region-$r$ dialogues, and allocation change is the switching rate
$
\Delta x^k=N^{-1}\sum_i\mathbf1(x_i^{(k+1)H}\neq x_i^{kH})\in[0,1].
$
The suggestion augments the next government observation
$
\widetilde o_r^{k+1}
=\left(
o_r^{k+1},
\Delta\boldsymbol\theta_{r,\mathrm{LLM}}^{k+1}
\right).
$
Regional actors use these augmented observations and share the realized government objective during policy learning. The suggestion is auxiliary: the learned government policy still produces the formal, feasibility-projected adjustment. Thus, the pipeline verifies evidence, explains conditions, updates cognition, presents feasible options, and returns unresolved concerns for the next policy cycle.

\subsection{Self-Evolving Skill Lifecycle}
The skill library follows four stages: candidate construction; deployment and trace collection; outcome-based evaluation and revision; and promotion with redundancy pruning. For traces $\mathfrak B_\ell^k$ involving skill $\ell$ through cycle $k$, we evaluate
\begin{equation}
\widehat J_\ell^k
=\operatorname{Eval}
\left(\mathfrak B_\ell^k;
\Delta\boldsymbol\Psi,\Delta x,\Delta W^G,\bar V\right)\in[0,1],
\end{equation}
where $\operatorname{Eval}$ is the normalized trace-scoring operator and $\bar V$ aggregates verification failures. Validated skills with consistently positive evidence form the promotion set $\mathcal P^k$; uncertain or context-mismatched skills are revised in the temporary pool, and ineffective or redundant skills are pruned:
\begin{equation}
\mathcal L^{k+1}
=\operatorname{Prune}_{K_{\mathcal L}}
\left(\mathcal L^k\cup\mathcal P^k\right),
\end{equation}
where $\operatorname{Prune}_{K_{\mathcal L}}$ retains at most $K_{\mathcal L}$ skills. Appendix~\ref{app:skill-evaluation} specifies the evaluation prompts, paired simulations, threshold, and update procedure. Skill evaluation measures factual validity, cognition and allocation changes, and marginal $W^G$, going beyond persuasion success or user satisfaction.
\begin{figure}[!t]
\centering
\begin{minipage}[t]{0.47\textwidth}
\vspace{0pt}
\centering
\captionof{table}{City-level alignment between the entry-friendliness benchmark and model OPC utility.}
\label{tab:preliminary_city_alignment}
\vspace{1pt}
\renewcommand{\arraystretch}{1.12}
\small
\begin{tabular*}{\linewidth}{@{\extracolsep{\fill}}lcccc@{}}
\hline
\textbf{City} & \textbf{External proxy} & $\mathbf{W}^{\mathbf{opc}}$ & \textbf{Ext. rank} & \textbf{Model rank} \\
\hline
\textbf{Shenzhen}  & 3.476 & 0.748 & 1 & 1 \\
\textbf{Beijing}   & 3.333 & 0.680 & 2 & 2 \\
\textbf{Shanghai}  & 3.231 & 0.621 & 3 & 3 \\
\textbf{Nanjing}   & 3.150 & 0.616 & 4 & 4 \\
\textbf{Chongqing} & 3.111 & 0.153 & 5 & 6 \\
\textbf{Hefei}     & 2.571 & 0.323 & 6 & 5 \\
\hline
\end{tabular*}
\end{minipage}\hfill
\begin{minipage}[t]{0.47\textwidth}
\vspace{5pt}
\centering
\includegraphics[width=\linewidth,trim=0 0 0 4pt,clip]{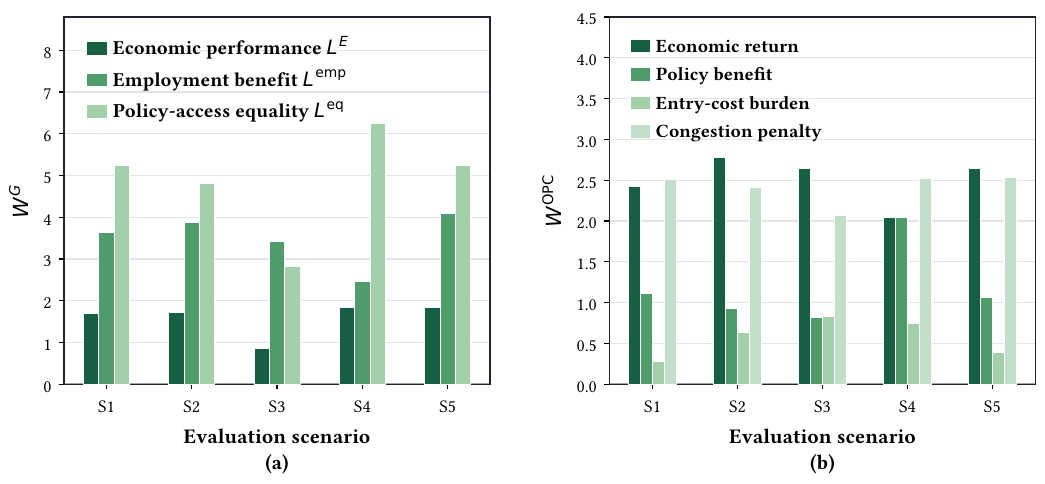}
\captionof{figure}{Government and OPC outcome decomposition: (a) components of $W^G$; (b) benefit and cost magnitudes in $W^{\mathrm{OPC}}$.}
\label{fig:welfare_components}
\end{minipage}
\end{figure}

\section{Experiments}
We first validate simulated entry patterns, then compare welfare across policies and OPC-simulator backbones, and finally examine ablations, mechanisms, a case study, and interactive judgments.

\subsection{Experimental Setup}
We evaluate PolicyMediator in an agent-based environment spanning six Chinese regions, five OPC types, six policy categories, and five scenarios: baseline conditions (S1), higher uncovered entry costs (S2), lower policy redemption (S3), greater OPC mobility (S4), and a technology-sector shock (S5). These settings, together with industrial fit, entry difficulty, capacity, and baseline returns, are derived from publicly available Chinese government policies, community tags, and socioeconomic statistics.

We vary the OPC-simulator backbone among Qwen3.5-9B, Ministral 3-8B, and Gemma 4-12B, while fixing Qwen for mediation, explanation, and consensus modules to isolate robustness to simulated OPC behavior.
We use four baseline categories which are nonadaptive controls (No Policy and Static Welfare); online learning (Contextual Bandit~\cite{li2010contextual,huang2026policysim} and Carbon-PPO~\cite{wang2024carbon}); objective-guided policy optimization (Nash Welfare~\cite{zhang2024noregret} and Belief-Robust~\cite{ciullo2023birdm}); and consensus formation (Bi-level Consensus~\cite{tu2024bilevel} and CONSENSAGENT~\cite{pitre2025consensagent}). 
Policy-learning baselines except No Policy share a frozen generic explainer  without government feedback or skill evolution, whereas consensus methods retain their native mechanisms. We report the aggregate government objective $W^G$, comprising economic performance, employment benefit, and policy-access equality, together with average realized OPC utility. Appendices A and C.1--C.2 provide data construction, implementation, and baseline details.

\subsection{Preliminary Study}
\label{sec:preliminary-study}
Before comparing complete systems, we first examine whether simulated OPC outcomes yield a plausible ranking of regional attractiveness. Holding baseline policies fixed, we record OPC entry shares and utility levels across six cities and compare the resulting city ranking with an descriptive entry-friendliness benchmark constructed from crawled public data. As shown in Table~\ref{tab:preliminary_city_alignment}, only the final adjacent pair swaps positions, yielding strong overall rank agreement (Spearman $\rho=0.943$, $p=0.00167$). This result indicates that the simulated payoffs preserve the cross-city gradient. Appendix~\ref{app:data-construction} provides the data sources and additional analyses.

\subsection{Overall Performance}
We compare methods on $W^G$ and OPC utility, then decompose $W^G$ into economic performance, employment benefit, policy-access equality, and their underlying components.
As shown in Table~\ref{tab:joint_welfare_comparison}, PolicyMediator ranks first on both metrics in all 15 settings, averaging $9.94$--$9.97$ and $0.719$--$0.720$, respectively, and exceeding the eight-baseline means by $42.4\%$ and $40.9\%$, with Carbon-PPO as the strongest baseline. Bi-level Consensus and CONSENSAGENT achieve only $7.54$ and $5.86$--$7.03$ in $W^G$, indicating that consensus alone cannot close the policy-feedback loop.

Figure~\ref{fig:welfare_components} decomposes government and OPC outcomes. Policy-access equality is largest $W^G$ component in S1, S2, S4, and S5, whereas employment benefit dominates S3. For $W^{\mathrm{OPC}}$, economic return is the largest benefit in four scenarios, while S4 exhibits nearly equal return and policy benefit; congestion penalty is the largest cost throughout. Lower redemption in S3 produces the smallest policy benefit and largest entry-cost burden. Overall, government outcomes are driven mainly by access and employment, whereas OPC utility reflects the balance between benefits and costs.

\FloatBarrier

\newcommand{\skillMaturationFigure}{%
\begin{figure}[!t]
  \centering
\includegraphics[width=\columnwidth,trim=0 0 0 0pt,clip]{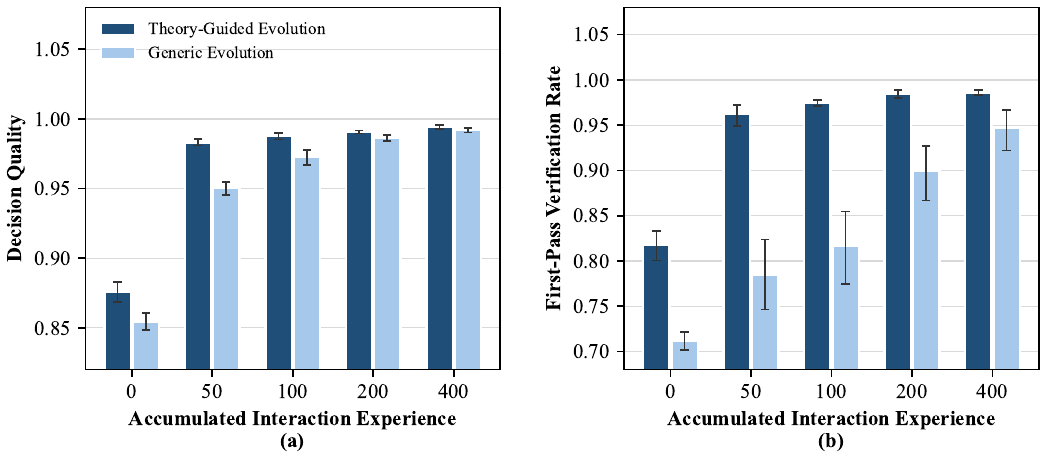}
  \caption{Skill evolution: (a) decision quality and (b) first-pass verification rate under theory-guided and generic evolution.}
  \Description{A single-column, two-panel grouped-bar figure comparing theory-guided and generic skill evolution over five experience checkpoints.}
  \label{fig:skill_maturation}
  \vspace{2pt}
\end{figure}%
}

\begin{figure*}[!t]
  \centering
\includegraphics[width=0.96\textwidth]{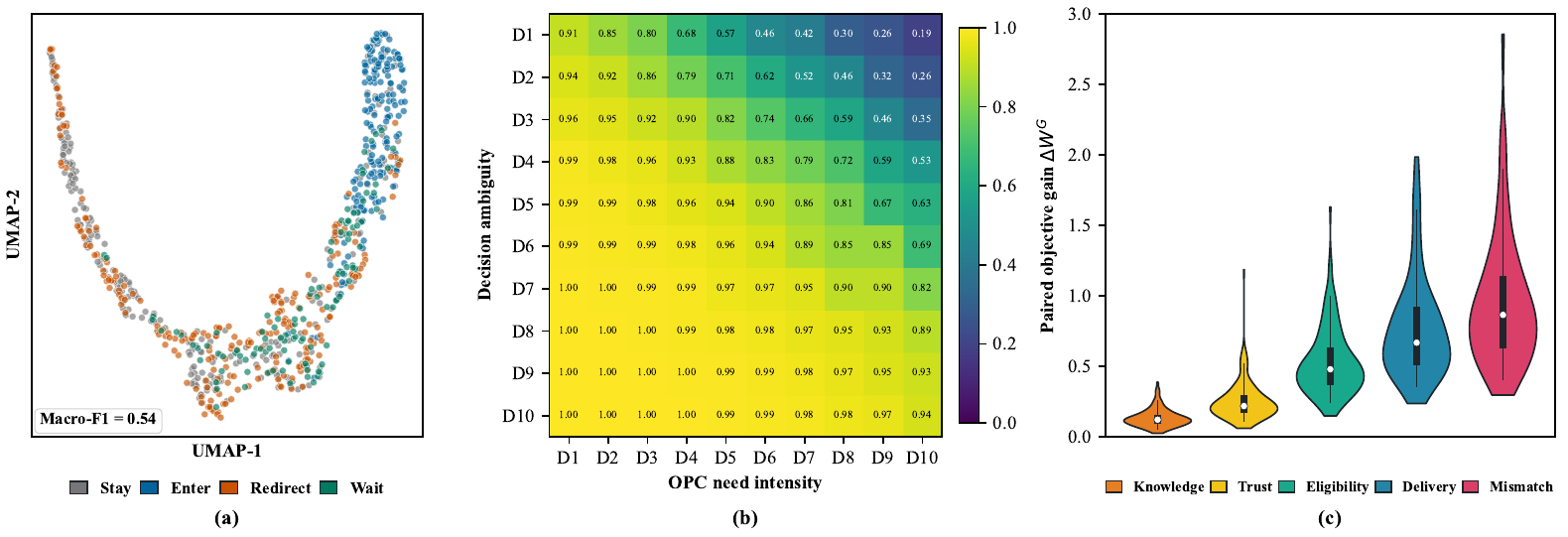}
  \caption{Mechanism analyses: (a) cognition changes by OPC transition, (b) feedback utilization across need intensity and decision ambiguity, and (c) paired $\Delta W^G$ by feedback type.}
  \Description{A cognitive-response embedding, a feedback-utilization heat map, and violin plots of paired $W^G$ gains.}
  \label{fig:mechanism_analysis}
\end{figure*}

\begin{table*}[!t]
\centering
\caption{Blinded pairwise interactive evaluation (\%; 100 pairs/comparison; $\ddagger$: exact binomial $p<.01$).}
\label{tab:interactive_evaluation}
\small
\setlength{\tabcolsep}{4pt}
\begin{tabular}{lcccccccccc|cc}
\toprule
\multirow{2}{*}{\textbf{Comparison}}
& \multicolumn{2}{c}{\textbf{Restraint}}
& \multicolumn{2}{c}{\textbf{Regional}}
& \multicolumn{2}{c}{\textbf{Targeting}}
& \multicolumn{2}{c}{\textbf{Specificity}}
& \multicolumn{2}{c}{\textbf{Progression}}
& \multicolumn{2}{|c}{\textbf{Average}} \\
\cmidrule(lr){2-3}
\cmidrule(lr){4-5}
\cmidrule(lr){6-7}
\cmidrule(lr){8-9}
\cmidrule(lr){10-11}
\cmidrule(lr){12-13}
& \textbf{Win} & \textbf{Lose}
& \textbf{Win} & \textbf{Lose}
& \textbf{Win} & \textbf{Lose}
& \textbf{Win} & \textbf{Lose}
& \textbf{Win} & \textbf{Lose}
& \textbf{Win} & \textbf{Lose} \\
\midrule
\textbf{PolicyMediator vs. CONSENSAGENT}
& 56.0 & 42.0
& 58.0 & 41.0
& $\mathbf{73.0}^{\ddagger}$ & 20.0
& 45.0 & 51.0
& $\mathbf{69.0}^{\ddagger}$ & 12.0
& $\mathbf{60.2}$ & 33.2 \\
\textbf{PolicyMediator vs. Contextual Bandit}
& $\mathbf{87.0}^{\ddagger}$ & 12.0
& $\mathbf{47.0}$ & 46.0
& 40.0 & 53.0
& $\mathbf{69.0}^{\ddagger}$ & 26.0
& $\mathbf{41.0}$ & 36.0
& $\mathbf{56.8}$ & 34.6 \\
\textbf{PolicyMediator vs. Generic Explainer}
& $\mathbf{92.0}^{\ddagger}$  & 4.0
& $\mathbf{81.0}^{\ddagger}$  & 17.0
& $\mathbf{100.0}^{\ddagger}$ & 0.0
& $\mathbf{91.0}^{\ddagger}$  & 7.0
& $\mathbf{100.0}^{\ddagger}$ & 0.0
& $\mathbf{92.8}$ & 5.6 \\
\bottomrule
\end{tabular}
\end{table*}
\skillMaturationFigure

\subsection{Further Analysis}
\textbf{Ablation Analysis.}
To identify PolicyMediator's key contributors, we conduct ablations of mediation, cognitive updating, policy feedback, skill evolution, and theory-guided selection. Table~\ref{tab:ablation} shows the full model achieves the highest five-scenario mean ($9.97$). Removing policy feedback causes the largest decline ($2.31$), highlighting its role in policy learning. Removing mediation follows ($2.19$), then cognitive updating ($1.86$), skill evolution ($1.24$), and theory guidance ($0.95$). These results identify feedback and mediation as primary contributors, with skill mechanisms providing additional gains.

\textbf{Cognitive-State and Mediation Analysis.} These tests examine whether cognition predicts OPC decisions, when mediated feedback helps, and which concerns offer policy leverage. Figure~\ref{fig:mechanism_analysis}(a) projects pre-aggregation cognition changes from 800 interactions, colored by the next OPC transition. A five-fold cross-validated classifier achieves macro-F1 $0.54$; despite partial overlap, the four distinguishable regions indicate predictive structure. For Figure~\ref{fig:mechanism_analysis}(b), need intensity denotes demand-shock magnitude, decision ambiguity denotes noise in the aggregate government signal, and feedback utilization is the recovered objective-gap fraction, $(W_{\mathrm{med}}-W_{\mathrm{base}})/(W_{\mathrm{best}}-W_{\mathrm{base}})$. Mean utilization is $0.842$; under low ambiguity it falls from $0.91$ to $0.19$ as need rises, but remains $0.94$--$1.00$ under the highest ambiguity. Figure~\ref{fig:mechanism_analysis}(c) compares 800 paired episodes: median $\Delta W^G$ ranges from regional mismatch ($0.865$) and delivery ($0.668$) to eligibility ($0.479$), trust ($0.218$), and knowledge ($0.121$). Thus, cognition informs OPC decisions, while government feedback is most valuable under ambiguous signals and actionable concerns.

\textbf{Skill Evolution Analysis.} This simulation compares theory-guided and generic evolution from 14 manually initialized skills with matched policies, contexts, and budgets. On held-out contexts, decision quality is one minus the mean utility-normalized regret of the recommended region, while first-pass verification is the fraction of recommendations passing the simulated check without revision. Figure~\ref{fig:skill_maturation} reports 20-seed means and 95\% confidence intervals over 0--400 experiences. At 400 experiences, decision quality is nearly identical, whereas verification remains higher. Thus, theory guidance enables faster skill maturation, stronger verification reliability, and near-optimal final decision quality.

\textbf{Case Study Analysis.} Table~\ref{tab:policymediator_case} in Appendix shows how bilateral mediation links an OPC concern to an immediate decision and later policy feedback. Despite a 34\% subsidy, 58\% computing support, and 12 rent-free months, the OPC questions whether computing support is a deliverable resource rather than a paper promise. PolicyMediator supplies redemption records and flags the unverified no-penalty clause. Trust falls from $.42$ to $.32$ and returns to $.42$ after verification, while awareness rises from $.45$ to $.56$. The OPC conditionally accepts at 75\% confidence, applies for computing and space support, and defers the loan until its legal clause is verified.

\subsection{Interactive Evaluation}

The interactive evaluation assesses mediation-response quality across decision-relevant dimensions. Ten trained evaluators---two OPC practitioners, four human--computer interaction or policy-modeling experts, and four graduate students---made blinded pairwise judgments on randomized outputs; Appendix~C.1 details the protocol and training. Restraint measures avoidance of pressure or unsupported promises; regional sensitivity measures use of region-specific evidence; targeting measures response to the stated concern. Specificity requires concrete, verifiable conditions or next steps, while progression requires non-repetitive resolution across turns. Table~\ref{tab:interactive_evaluation} reports wins and losses, omitting ties. PolicyMediator wins 13 of 15 criterion-level comparisons, including all five against the Generic Explainer. It trails only CONSENSAGENT in specificity and Contextual Bandit in targeting, indicating that evaluators broadly prefer its responses across the interactive criteria.

\section{Conclusion, Limitations and Future Work}
We presented PolicyMediator for cross-regional policy learning when incomplete policy information causes OPCs to reject otherwise suitable regions. It combines cognition-constrained Stackelberg learning, verified bilateral mediation, and outcome-aware skill evolution to distinguish cognitive barriers from policy mismatch and convert OPC concerns into policy feedback. This closed loop improves regional matching, economic performance, employment benefits, and policy-access equality beyond disconnected policy optimization. 
The results apply to coordinated-government settings with verifiable policies, comparable regional options, and observable entry and capacity signals. Limitations include simulated OPC responses and cooperative-government assumptions, restricting generalization to competitive jurisdictions and long-term behavior. Future work will consider noncooperative regions, extend to broader macroeconomic settings, and further investigate welfare outcomes and trade-offs.

\bibliographystyle{ACM-Reference-Format}
\bibliography{PolicyMediator_refs}

@article{alyakoob2026market,
  author  = {Mohammed Alyakoob and Mohammad Rahman},
  title   = {Market Design Choices, Racial Discrimination, and Equitable Microentrepreneurship in Digital Marketplaces},
  journal = {Management Science},
  year    = {2026},
  volume  = {72},
  number  = {3},
  pages   = {1878--1903},
}

@article{guo2025skill,
  author  = {Xue Guo and Zhi Cheng and Paul A. Pavlou},
  title   = {Skill-Biased Technical Change, Again? Online Gig Platforms and Local Employment},
  journal = {Information Systems Research},
  year    = {2025},
  volume  = {36},
  number  = {3},
  pages   = {1354--1374},
}

@article{jones2024gendered,
  author  = {Sam Jones and Kunal Sen},
  title   = {The Gendered Impact of Digital Jobs Platforms: Experimental Evidence from Mozambique},
  journal = {The World Bank Economic Review},
  year    = {2025},
  volume  = {39},
  number  = {1},
  pages   = {143--163},
  doi = {10.1093/wber/lhae019},
}

@article{doggett2024digital,
  author  = {Olivia Doggett and Sarah Gram and Sophia Jit and Basel Harb and Houssam Harb and Mohamad Shafik and Robert Soden and Matt Ratto},
  title   = {Digital Gig Work Under Unstable Energy Infrastructures: Invisible Workarounds and Alternative Imaginaries},
  journal = {Proceedings of the ACM on Human-Computer Interaction},
  year    = {2024},
  volume  = {8},
  number  = {CSCW2},
  pages   = {1--26},
}

@article{leong2024coordination,
  author  = {Carmen Leong and Silvia Lin and Felix Tan and Jie Yu},
  title   = {Coordination in a Digital Platform Organization},
  journal = {Information Systems Research},
  year    = {2024},
  volume  = {35},
  number  = {1},
  pages   = {363--393},
}

@article{agarwal2024mobile,
  author  = {Aparajita Agarwal and Valentina A. Assenova},
  title   = {Mobile Money as a Stepping Stone to Financial Inclusion: How Digital Multisided Platforms Fill Institutional Voids},
  journal = {Organization Science},
  year    = {2024},
  volume  = {35},
  number  = {3},
  pages   = {769--787},
}

@article{liu2024unintended,
  author  = {Yi Liu and Bowen Lou and Xinyi Zhao and Xinxin Li},
  title   = {Unintended Consequences of Advances in Matching Technologies: Information Revelation and Strategic Participation on Gig-Economy Platforms},
  journal = {Management Science},
  year    = {2024},
  volume  = {70},
  number  = {3},
  pages   = {1729--1754},
}

@article{farmer2009economy,
  author  = {J. Doyne Farmer and Duncan Foley},
  title   = {The Economy Needs Agent-Based Modelling},
  journal = {Nature},
  year    = {2009},
  volume  = {460},
  number  = {7256},
  pages   = {685--686},
}

@inproceedings{park2023generative,
  author    = {Joon Sung Park and Joseph O'Brien and Carrie Jun Cai and Meredith Ringel Morris and Percy Liang and Michael S. Bernstein},
  title     = {Generative Agents: Interactive Simulacra of Human Behavior},
  booktitle = {Proceedings of the 36th Annual ACM Symposium on User Interface Software and Technology},
  year      = {2023},
  pages     = {1--22},
}

@inproceedings{li2024econagent,
  author    = {Nian Li and Chen Gao and Mingyu Li and Yong Li and Qingmin Liao},
  title     = {{EconAgent}: Large Language Model-Empowered Agents for Simulating Macroeconomic Activities},
  booktitle = {Proceedings of the 62nd Annual Meeting of the Association for Computational Linguistics (Volume 1: Long Papers)},
  year      = {2024},
  pages     = {15523--15536},
}

@inproceedings{papachristou2023dynamic,
  author    = {Marios Papachristou and Siddhartha Banerjee and Jon Kleinberg},
  title     = {Dynamic Interventions for Networked Contagions},
  booktitle = {Proceedings of the ACM Web Conference 2023},
  year      = {2023},
  pages     = {3519--3529},
}

@inproceedings{yong2026intelliplanner,
  author    = {Jiarui Yong and Chaoyi Sun and Youfang Wang and Xiang Zhou},
  title     = {{Intelli-Planner}: Towards Customized Urban Planning via Large Language Model Empowered Reinforcement Learning},
  booktitle = {Proceedings of the ACM Web Conference 2026},
  year      = {2026},
  pages     = {9385},
}

@inproceedings{bai2021sample,
  author    = {Yu Bai and Chi Jin and Huan Wang and Caiming Xiong},
  title     = {Sample-Efficient Learning of Stackelberg Equilibria in General-Sum Games},
  booktitle = {Advances in Neural Information Processing Systems},
  year      = {2021},
  volume    = {34},
  pages     = {25799--25811},
}

@article{zheng2022aieconomist,
  author  = {Stephan Zheng and Alexander Trott and Sunil Srinivasa and David C. Parkes and Richard Socher},
  title   = {The {AI} Economist: Taxation Policy Design via Two-Level Deep Multiagent Reinforcement Learning},
  journal = {Science Advances},
  year    = {2022},
  volume  = {8},
  number  = {18},
  pages   = {eabk2607},
}

@inproceedings{huang2026policysim,
  author    = {Renhong Huang and Ning Tang and Jiarong Xu and Yuxuan Cao and Qingqian Tu and Sheng Guo and Bo Zheng and Huiyuan Liu and Yang Yang},
  title     = {{PolicySim}: An {LLM}-Based Agent Social Simulation Sandbox for Proactive Policy Optimization},
  booktitle = {Proceedings of the ACM Web Conference 2026},
  year      = {2026},
  pages     = {4781--4792},
}

@inproceedings{li2010contextual,
  author    = {Lihong Li and Wei Chu and John Langford and Robert E. Schapire},
  title     = {A Contextual-Bandit Approach to Personalized News Article Recommendation},
  booktitle = {Proceedings of the 19th International Conference on World Wide Web},
  year      = {2010},
  pages     = {661--670},
  doi       = {10.1145/1772690.1772758},
}

@inproceedings{wang2024carbon,
  author    = {Han Wang and Wenhao Li and Hongyuan Zha and Baoxiang Wang},
  title     = {Carbon Market Simulation with Adaptive Mechanism Design},
  booktitle = {Proceedings of the Thirty-Third International Joint Conference on Artificial Intelligence},
  year      = {2024},
  pages     = {8824--8828},
  doi       = {10.24963/ijcai.2024/1041},
}

@article{ciullo2023birdm,
  author  = {Alessio Ciullo and Alessio Domeneghetti and Jan H. Kwakkel and Karin M. De Bruijn and Frans Klijn and Attilio Castellarin},
  title   = {Belief-Informed Robust Decision Making ({BIRDM}): Assessing Changes in Decision Robustness Due to Changing Distributions of Deep Uncertainties},
  journal = {Environmental Modelling \& Software},
  year    = {2023},
  volume  = {159},
  pages   = {105560},
  doi     = {10.1016/j.envsoft.2022.105560},
}

@inproceedings{aguilera2025population,
  author    = {Alba Aguilera and Miquel Albert{\'i} and Nardine Osman and Georgina Curto},
  title     = {Population Synthesis with Motivational Attributes: A Path Towards Cultural Variation in Agent-Based Models},
  booktitle = {Proceedings of the 28th European Conference on Artificial Intelligence},
  year      = {2025},
  pages     = {3751--3758},
}

@article{zhong2023stackelbergnash,
  author  = {Han Zhong and Zhuoran Yang and Zhaoran Wang and Michael I. Jordan},
  title   = {Can Reinforcement Learning Find Stackelberg--Nash Equilibria in General-Sum Markov Games with Myopically Rational Followers?},
  journal = {Journal of Machine Learning Research},
  year    = {2023},
  volume  = {24},
  number  = {35},
  pages   = {1--52},
}

@inproceedings{mi2024taxai,
  author    = {Qirui Mi and Siyu Xia and Yan Song and Haifeng Zhang and Shenghao Zhu and Jun Wang},
  title     = {{TaxAI}: A Dynamic Economic Simulator and Benchmark for Multi-Agent Reinforcement Learning},
  booktitle = {Proceedings of the 23rd International Conference on Autonomous Agents and Multiagent Systems},
  year      = {2024},
  pages     = {1390--1399},
}

@inproceedings{zhang2024noregret,
  author    = {Mengxiao Zhang and Ramiro Deo-Campo Vuong and Haipeng Luo},
  title     = {No-Regret Learning for Fair Multi-Agent Social Welfare Optimization},
  booktitle = {Advances in Neural Information Processing Systems},
  year      = {2024},
  volume    = {37},
  pages     = {57671--57700},
}

@inproceedings{bindel2011opinion,
  author    = {David Bindel and Jon Kleinberg and Sigal Oren},
  title     = {How Bad Is Forming Your Own Opinion?},
  booktitle = {Proceedings of the 52nd Annual IEEE Symposium on Foundations of Computer Science},
  year      = {2011},
  pages     = {57--66},
}

@article{auletta2019consensus,
  author  = {Vincenzo Auletta and Angelo Fanelli and Diodato Ferraioli},
  title   = {Consensus in Opinion Formation Processes in Fully Evolving Environments},
  journal = {Proceedings of the AAAI Conference on Artificial Intelligence},
  year    = {2019},
  volume  = {33},
  number  = {1},
  pages   = {6022--6029},
}

@inproceedings{musco2018polarization,
  author    = {Cameron Musco and Christopher Musco and Charalampos E. Tsourakakis},
  title     = {Minimizing Polarization and Disagreement in Social Networks},
  booktitle = {Proceedings of the 2018 World Wide Web Conference},
  year      = {2018},
  pages     = {369--378},
}

@article{wall2001mediation,
  author  = {James A. Wall and John B. Stark and Rhetta L. Standifer},
  title   = {Mediation: A Current Review and Theory Development},
  journal = {Journal of Conflict Resolution},
  year    = {2001},
  volume  = {45},
  number  = {3},
  pages   = {370--391},
}

@incollection{susskind2018deliberative,
  author    = {Lawrence Susskind and Jessica Gordon and Yasmin Zaerpoor},
  title     = {Deliberative Democracy and Public Dispute Resolution},
  booktitle = {The Oxford Handbook of Deliberative Democracy},
  year      = {2018},
  pages     = {710--727},
}

@article{dalkin2015mechanism,
  author  = {Sonia Michelle Dalkin and Joanne Greenhalgh and Diana Jones and Bill Cunningham and Monique Lhussier},
  title   = {What's in a Mechanism? Development of a Key Concept in Realist Evaluation},
  journal = {Implementation Science},
  year    = {2015},
  volume  = {10},
  pages   = {49},
}

@article{tessler2024common,
  author  = {Michael Henry Tessler and Michiel A. Bakker and Daniel Jarrett and Hannah Sheahan and Martin J. Chadwick and Raphael Koster and Georgina Evans and Lucy Campbell-Gillingham and Tantum Collins and David C. Parkes and Matthew Botvinick and Christopher Summerfield},
  title   = {{AI} Can Help Humans Find Common Ground in Democratic Deliberation},
  journal = {Science},
  year    = {2024},
  volume  = {386},
  number  = {6719},
  pages   = {eadq2852},
}

@article{chen2026mediation,
  author  = {Junjie Chen and Haitao Li and Minghao Qin and Yujia Zhou and Yanxue Ren and Wuyue Wang and Yiqun Liu and Yueyue Wu and Qingyao Ai},
  title   = {Simulating Dispute Mediation with {LLM}-Based Agents for Legal Research},
  journal = {Proceedings of the AAAI Conference on Artificial Intelligence},
  year    = {2026},
  volume  = {40},
  number  = {35},
  pages   = {29368--29375},
}

@inproceedings{pitre2025consensagent,
  author    = {Priya Pitre and Naren Ramakrishnan and Xuan Wang},
  title     = {{CONSENSAGENT}: Towards Efficient and Effective Consensus in Multi-Agent {LLM} Interactions Through Sycophancy Mitigation},
  booktitle = {Findings of the Association for Computational Linguistics: ACL 2025},
  year      = {2025},
  pages     = {22112--22133},
}

@article{tu2024bilevel,
  author  = {Yan Tu and Jiajia Song and Yutong Xie and Xiaoyang Zhou and Benjamin Lev},
  title   = {Facilitating Large-Scale Group Decision-Making in Social Networks: A Bi-Level Consensus Model with Social Influence},
  journal = {Information Fusion},
  year    = {2024},
  volume  = {105},
  pages   = {102258},
  doi     = {10.1016/j.inffus.2024.102258},
}

@article{sutton1999options,
  author  = {Richard S. Sutton and Doina Precup and Satinder Singh},
  title   = {Between {MDPs} and Semi-{MDPs}: A Framework for Temporal Abstraction in Reinforcement Learning},
  journal = {Artificial Intelligence},
  year    = {1999},
  volume  = {112},
  number  = {1--2},
  pages   = {181--211},
}

@inproceedings{eysenbach2019diayn,
  author    = {Benjamin Eysenbach and Abhishek Gupta and Julian Ibarz and Sergey Levine},
  title     = {Diversity Is All You Need: Learning Skills Without a Reward Function},
  booktitle = {International Conference on Learning Representations},
  year      = {2019},
}

@article{wang2024continual,
  author  = {Liyuan Wang and Xingxing Zhang and Hang Su and Jun Zhu},
  title   = {A Comprehensive Survey of Continual Learning: Theory, Method and Application},
  journal = {IEEE Transactions on Pattern Analysis and Machine Intelligence},
  year    = {2024},
  volume  = {46},
  number  = {8},
  pages   = {5362--5383},
}

@inproceedings{shinn2023reflexion,
  author    = {Noah Shinn and Federico Cassano and Ashwin Gopinath and Karthik Narasimhan and Shunyu Yao},
  title     = {Reflexion: Language Agents with Verbal Reinforcement Learning},
  booktitle = {Advances in Neural Information Processing Systems},
  year      = {2023},
  volume    = {36},
  pages     = {8634--8652},
}

@inproceedings{madaan2023selfrefine,
  author    = {Aman Madaan and Niket Tandon and Prakhar Gupta and Skyler Hallinan and Luyu Gao and Sarah Wiegreffe and Uri Alon and Nouha Dziri and Shrimai Prabhumoye and Yiming Yang and Shashank Gupta and Bodhisattwa Prasad Majumder and Katherine Hermann and Sean Welleck and Amir Yazdanbakhsh and Peter Clark},
  title     = {Self-Refine: Iterative Refinement with Self-Feedback},
  booktitle = {Advances in Neural Information Processing Systems},
  year      = {2023},
  volume    = {36},
  pages     = {46534--46594},
}

@inproceedings{zhang2024agentcf,
  author    = {Junjie Zhang and Yupeng Hou and Ruobing Xie and Wenqi Sun and Julian McAuley and Wayne Xin Zhao and Leyu Lin and Ji-Rong Wen},
  title     = {{AgentCF}: Collaborative Learning with Autonomous Language Agents for Recommender Systems},
  booktitle = {Proceedings of the ACM Web Conference 2024},
  year      = {2024},
  pages     = {3679},
}

@inproceedings{yu2026polyskill,
  author    = {Simon Yu and Gang Li and Weiyan Shi and Peng Qi},
  title     = {{PolySkill}: Learning Generalizable Skills Through Polymorphic Abstraction for Continual Learning},
  booktitle = {International Conference on Learning Representations},
  year      = {2026},
}

@inproceedings{yang2026evods,
  author    = {Zherui Yang and Fan Liu and Yansong Ning and Hao Liu},
  title     = {{EvoDS}: Self-Evolving Autonomous Data Science Agent with Skill Learning and Context Management},
  booktitle = {Proceedings of the 32nd ACM SIGKDD Conference on Knowledge Discovery and Data Mining},
  year      = {2026},
  pages     = {6128--6139},
}

@inproceedings{wang2026skilllibrary,
  author    = {Jiongxiao Wang and Qiaojing Yan and Yawei Wang and Yijun Tian and Soumya Smruti Mishra and Zhichao Xu and Megha Gandhi and Panpan Xu and Lin Lee Cheong},
  title     = {Reinforcement Learning for Self-Improving Agent with Skill Library},
  booktitle = {Proceedings of the 64th Annual Meeting of the Association for Computational Linguistics (Volume 1: Long Papers)},
  year      = {2026},
  pages     = {1529--1550},
}

@inproceedings{jiao2026agentic,
  author    = {Zhengbo Jiao and Shaobo Wang and Zifan Zhang and Xuan Ren and Wei Wang and Bing Zhao and Wei Hu and Linfeng Zhang},
  title     = {Agentic Proposing: Enhancing Large Language Model Reasoning via Compositional Skill Synthesis},
  booktitle = 	 {Proceedings of the 43rd International Conference on Machine Learning},
  pages = 	 {53778--53799},
  year = 	 {2026},
  volume = 	 {306},
}

@article{gefen2003trust,
  author  = {David Gefen and Elena Karahanna and Detmar W. Straub},
  title   = {Trust and {TAM} in Online Shopping: An Integrated Model},
  journal = {MIS Quarterly},
  year    = {2003},
  volume  = {27},
  number  = {1},
  pages   = {51--90},
  doi     = {10.2307/30036519},
}

@article{rosenthal1973congestion,
  author  = {Robert W. Rosenthal},
  title   = {A Class of Games Possessing Pure-Strategy Nash Equilibria},
  journal = {International Journal of Game Theory},
  year    = {1973},
  volume  = {2},
  number  = {1},
  pages   = {65--67},
  doi     = {10.1007/BF01737559},
}

@article{petty1986elm,
  author  = {Richard E. Petty and John T. Cacioppo},
  title   = {The Elaboration Likelihood Model of Persuasion},
  journal = {Advances in Experimental Social Psychology},
  year    = {1986},
  volume  = {19},
  pages   = {123--205},
  doi     = {10.1016/S0065-2601(08)60214-2},
}

@misc{unitednations2015agenda,
  author       = {{United Nations}},
  title        = {Transforming Our World: The 2030 Agenda for Sustainable Development},
  year         = {2015},
  howpublished = {United Nations General Assembly Resolution 70/1},
  url          = {https://sdgs.un.org/2030agenda},
}

@misc{cgss2023,
  author       = {{CGSS Project Team}},
  title        = {Chinese General Social Survey ({CGSS2023})},
  year         = {2026},
  howpublished = {Renmin University of China Social Science Data Platform},
  url          = {https://hdl.handle.net/20.500.14796/LWP5X7},
  note         = {Dataset},
}

@article{ghadimi2013zerothorder,
  author  = {Saeed Ghadimi and Guanghui Lan},
  title   = {Stochastic First- and Zeroth-Order Methods for Nonconvex Stochastic Programming},
  journal = {SIAM Journal on Optimization},
  year    = {2013},
  volume  = {23},
  number  = {4},
  pages   = {2341--2368},
  doi     = {10.1137/120880811},
}

@article{nesterov2017gradientfree,
  author  = {Yurii Nesterov and Vladimir Spokoiny},
  title   = {Random Gradient-Free Minimization of Convex Functions},
  journal = {Foundations of Computational Mathematics},
  year    = {2017},
  volume  = {17},
  number  = {2},
  pages   = {527--566},
  doi     = {10.1007/s10208-015-9296-2},
}

@article{humphries2020information,
  title={Information frictions and access to the Paycheck Protection Program},
  author={Humphries, John Eric and Neilson, Christopher A and Ulyssea, Gabriel},
  journal={Journal of public economics},
  volume={190},
  pages={104244},
  year={2020},
  publisher={Elsevier}
}

@inproceedings{kim2026generative,
  title={How Generative AI Enables Solo Entrepreneurship While Preserving Team Advantages},
  author={Kim, Hyunso and Kang, Hyo and Song, Jaeyong},
  booktitle={Academy of Management Proceedings},
  volume={2026},
  number={1},
  pages={16114},
  year={2026},
  
}

\appendix
\section*{Appendix}
\section{Data Construction}
\label{app:data-construction}
\subsection{Sources and use}

\textbf{Observed inputs.} Guided by the multi-regional government report, we coded six policy dimensions from 134 policies and 617 communities; a reference library contains 106 policies and 96 communities. Records were crawled and cleaned from OPCquan\footnote{\url{https://www.opcquan.com/}} and other public OPC portals\footnote{\url{https://github.com/siuserxiaowei/opc-policy}} . Policy fields provide evidence, while community tags indicate industry, infrastructure, and capacity. CGSS2023 supplies population, employment, wage, and cognition inputs~\cite{cgss2023}. The preliminary entry-friendliness proxy is derived from crawled public OPC community records, combining normalized city-level community availability, disclosed entry accessibility, and service coverage. It is used as a descriptive rank-order reference for Table~\ref{tab:preliminary_city_alignment}. Our code is built based on LLMEconomist\footnote{\url{https://github.com/sethkarten/LLMEconomist}} .

\textbf{Derived and controlled variables.} For each region and category, observed policy intensity $\widehat\theta_{rj}^{\mathrm{obs}}$ combines matching-record coverage and keyword evidence; infrastructure aggregates operator coverage and disclosed workstations, and entry difficulty $d_r$ decreases with this infrastructure proxy. Regional capacity $K_r$ preserves each region's share of disclosed workstations when mapped to the synthetic population. Community industry tags determine the OPC-type--industry mapping $\Phi_{lg}$; combined with national microenterprise employment $\widehat s_g$ and regional wages $\bar w_{rg}$, it gives the resident-income proxy $\widehat R_{lr}=\sum_g\Phi_{lg}\widehat s_g\bar w_{rg}$. 
The final compatibility matrix $M=(m_{lr})$, policy-sensitivity matrix
$\Gamma$, baseline returns $v_r$, type proportions, initial cognition, and
expenditure reference are fixed before evaluation and shared by all methods;
they are not re-estimated for individual baselines. Resident outcomes are
represented by aggregate employment and wage proxies.
We define the fiscal-cost proxy as
\(\widetilde C_r^{\mathrm{dir},k}
=\frac{\lambda_r\langle n_r^t\rangle_k}{6N}
\sum_{j=1}^{6}\!\left[\rho_r^k\theta_{rj}^k
+q(\rho_r^k\theta_{rj}^k)^2\right]\), where $\lambda_r$ is the region-specific expenditure scale and quadratic cost coefficient \(q=1\).

\subsection{Scales and normalization}
To ensure comparability across modeled variables, policy coverage and
infrastructure are min--max scaled across the six regions, with entry
difficulty defined as the reversed infrastructure score. Monetary and
employment quantities are divided by their dataset-level reference values,
while identifiers and counts remain unscaled.

\section{Prompts}
\label{app:prompts}
\subsection{Prompt Architecture}
The panels below summarize the prompt templates used by PolicyMediator, with
braces indicating fields instantiated at run time.

\promptpanel{OPC cognition and profile prompt}{
Regions and verified policies: \{region-wise records\}; communities: \{community records\}; industry directions: \{industry evidence\}; objective/scenario: \{configuration\}.\\
For each of five OPC types, estimate awareness $z_i\in[0,1]$, expected delivery (trust) $b_i\in[0,1]$, and region-specific relevance $(q_{ir})_{r\in\mathcal R}\in[0,1]^{|\mathcal R|}$. Return strict JSON only; do not add unsupported facts.}

\promptpanel{Mediator prompt}{
SYSTEM: You are a regional OPC policy-service mediator. Use only verified regional facts; do not invent amounts, eligibility, delivery, or promises. Use plain business language and at most 120 English words.\\
USER: Evidence=\{verified facts\}; fact-check dimensions=\{CMO checks\}; policy intensity=\{$\theta_1,\ldots,\theta_6$\}; OPC type=\{type\}; state=\{$\Psi_{ir}^{\mathrm{opc}}=(z_i,b_i,q_{ir})$\}; skill=\{skill specification\}; previous reply=\{reply\}; turn intent=\{intent\}. Explain conditions and offer feasible options.}

\promptpanel{OPC response prompt}{
SYSTEM: You are the founder of a \{type\} OPC choosing among regions. Be pragmatic and risk-aware, prioritize your own business, do not agree automatically, and raise one specific unresolved concern. Acknowledge verifiable information. Respond in first person in 1--2 sentences.\\
USER: Mediator message=\{message\}; business profile=\{profile\}; regional situation=\{state\}. On the final turn append exactly: STANCE=accept|conditional|decline and CONFIDENCE=0--100.}

\promptpanel{Skill specification prompt}{
Persuasion technique: \{technique\}\\
Dialog protocol: \{protocol\}\\
Must not: \{behavioral and evidence constraints\}\\
Apply only when \{OPC type, cognitive state, and concern conditions\} hold; otherwise select or revise another skill.}

\promptpanel{OPC cognition-change analysis}{
SYSTEM: Infer cognitive changes only from the transcript and verified evidence. Return raw update signals.\\
USER: Transcript=\{dialogue\}; prior state=\{$\boldsymbol\Psi_i=(z_i,b_i,(q_{ir})_{r\in\mathcal R})$\}. Return strict JSON with raw\_z, raw\_b, one raw\_q per region, and a brief rationale. Each raw signal is real-valued, with its sign indicating the direction of change. The system applies $\boldsymbol{\delta}_{\Psi}\odot\tanh(\cdot)$, where $\boldsymbol{\delta}_{\Psi}=(.15,.12,(.10)_{r\in\mathcal R})$, and projects the updated state to $[0,1]$.}

\subsection{Skill Library and Theoretical Routes}
The skill library operationalizes two complementary theories. Context--Mechanism--Outcome (CMO) analysis asks which mechanism works for which OPC and region: context records policy, capacity, fit, firm type, and cognition; mechanisms record clarification, disclosure, trust repair, or redirection; outcomes cover verification, cognition, choice, allocation, and $W^G$. Eight fact-check skills examine source, capacity, fit, amount, eligibility, timeline, risk, and norm evidence.
The Elaboration Likelihood Model (ELM) determines how verified information is communicated. Central-route skills provide detailed eligibility, benefit, procedure, fit, and capacity-risk reasoning when the OPC has sufficient relevance and processing readiness. Peripheral-route skills use concise source cues, trust repair, policy-access explanations, verifiable cases, and low-cost next steps when awareness, attention, or trust is limited. Route-switching skills move from inquiry to diagnosis, deepen or summarize explanations, and convert cases into explicit comparisons as the concern changes. Both routes obey the same evidence boundary; third-party audit or certification is mentioned only when supported by regional evidence.

Each skill record stores its target OPC types and cognitive states, applicability conditions, admissible evidence, technique, dialogue protocol, prohibited behavior, and evaluation history. It follows four stages: candidate construction, temporary deployment and trace collection, outcome-based evaluation and revision, and promotion with redundancy pruning. Evolution therefore changes reusable evidence-grounded skills and their applicability.

\subsection{Skill Evaluation and Update Protocol}
\label{app:skill-evaluation}
This protocol turns interactions into reusable knowledge without rewarding mediation success alone. Each invocation records context, route, evidence, dialogue, verification, cognition, and response; cycle completion adds allocation and $W^G$. Matched-context evaluation covers validity, concern resolution, recommendation quality, cognition, allocation, failures, and marginal $W^G$. Verified candidates above $\tau_J$ as evolution cycles are promoted; uncertain ones are revised and retested, while ineffective or redundant skills are removed under capacity $K_{\mathcal L}$. Cross-region reuse rechecks local evidence and capacity, and promotion requires repeated conditional performance.

\begin{table*}[!t]
\centering
\caption{Representative bilateral mediation case for a creator OPC. Here $q_{ir}$ is the region-level relevance, i.e., the mean of six policy-category perceptions for region $r$.}
\label{tab:policymediator_case}
\footnotesize
\setlength{\tabcolsep}{3pt}
\renewcommand{\arraystretch}{1.00}
\begin{tabularx}{\textwidth}{
@{}
>{\raggedright\arraybackslash}p{0.13\textwidth}
>{\raggedright\arraybackslash}X
>{\raggedright\arraybackslash}p{0.285\textwidth}
@{}}
\toprule
\multicolumn{3}{c}{\textbf{Simulated OPC Dialogue: Policy Explanation and Trust Building}}\\
\midrule
\textbf{Policy context} & \textbf{OPC type:} creator. \textbf{Policy package:} 34\% price subsidy, 58\% computing support, and a 12-month rent-free period. & Baseline $z=.45$; $b=.42$; $q_{i,\mathrm{Shenzhen}}=.475$. \\
\textbf{Turn 0} & \textbf{Mediator---concern elicitation:} ``Which step worries you most: reimbursement, approval, or something else? We will answer from verified policy facts.'' \par \textbf{OPC:} ``Is computing support actually usable or only a paper promise? Show me a signed quota or access agreement.'' & $z=.40$; $b=.32$; $q_{i,\mathrm{Shenzhen}}=.458$. Delivery uncertainty emerges; trust and regional relevance decline. \\
\textbf{Turn 1} & \textbf{Mediator---verified evidence:} ``The Commission publishes redemption records for computing vouchers and rent-free space. Vouchers cover actual usage, while guaranteed loans require guarantor approval.'' \par \textbf{OPC:} ``The records reduce voucher risk, but I still need the local approval rate for the guaranteed loan.'' & $z=.48$; $b=.37$; $q_{i,\mathrm{Shenzhen}}=.480$. Voucher evidence reduces delivery concerns, while loan approval remains uncertain. \\
\textbf{Turn 2} & \textbf{Mediator---procedural clarification:} ``Voucher and space support are documented, but the loan's no-penalty clause remains unverified. You may apply for verified support while seeking written loan terms.'' \par \textbf{OPC:} ``I will start those two applications but defer the loan until I verify the no-penalty clause.'' & $z=.56$; $b=.42$; $q_{i,\mathrm{Shenzhen}}=.495$. Verified information restores trust despite remaining loan uncertainty. \\
\textbf{Observed outcome} & Final stance: \textbf{Conditional acceptance}; confidence $=75\%$; mediation success: \textbf{True}. & Net update: $\Delta z=+.11$, $\Delta b=.00$, and $\Delta q_{i,\mathrm{Shenzhen}}=+.020$. Each per-turn update satisfies Eq.~\eqref{eq:cognitive-transition}. \\
\midrule
\multicolumn{3}{c}{\textbf{Bilateral Mediator Feedback Trace}}\\
\midrule
\textbf{OPC-facing route} & \textbf{Context:} Shenzhen policy service for a low-trust, moderately aware OPC. \par \textbf{Initial route:} Elicit delivery concerns before presenting verified policy evidence. & Concern elicitation identifies delivery risks without pressuring immediate entry. \\
\textbf{Route switching} & After the OPC requests verifiable records, mediation shifts from concern elicitation to central-route reasoning about voucher delivery and loan eligibility. & Evidence-based clarification addresses the stated concerns while preserving uncertainty. \\
\textbf{Government feedback} & \textbf{Mediator $\rightarrow$ government:} ``Public records address voucher delivery, but the loan's no-penalty clause remains unverified. Publish the legal terms before making claims.'' & The unresolved legal boundary, rather than the subsidy level, is returned to policy learning. \\
\textbf{Next-cycle advice} & Retain existing support and request verified loan terms before recommending further policy-intensity adjustments. & Advisory only; the regional actor makes formal updates. The trace is \emph{concern $\rightarrow$ verified response $\rightarrow$ government feedback $\rightarrow$ advice}. \\
\bottomrule
\end{tabularx}
\end{table*}

\section{Experiments}
\subsection{Implementation Details}
\label{app:zeroth-order-updates}
\textbf{Environment.} Experiments use Python 3.10.20, PyTorch 2.11.0/CUDA 13.0, vLLM 0.22.0, and Transformers 5.9.0 on four A100-80GB GPUs; an RTX A400 supports preprocessing. The three frozen OPC backbones are served sequentially; Ministral uses an offline bfloat16 conversion of its official FP8 checkpoint. FlashInfer sampling and Qwen thinking mode are disabled for reproducible JSON output.

\textbf{Interactive evaluation.}
Before annotation, evaluators completed common rubric training. Matched outputs were presented in randomized order without method or backbone identities, and all comparators received the same verified policy evidence. 
Ties were allowed and can be calculated based on Table 4 and evaluators were compensated above the local statutory minimum hourly wage.

\textbf{Zeroth-order policy learning.} Each region uses a lightweight deterministic actor, and all actors share the realized $W^G$. One Gaussian direction jointly perturbs all actor parameters; each $\pm$ pair starts from the same saved state and seed, while state and skill evolution between cycles contributes to $V_T$. The return difference updates each actor by Eq.~\eqref{eq:zeroth-order-policy-gradient}.
Table~\ref{tab:hyperparameters} summarizes the released settings. Compared methods use matched populations, initial states, feasible policy ranges, and random seeds.

\begin{table}[t]
\centering
\caption{Simulation and policy-learning hyperparameters.}
\label{tab:hyperparameters}
\footnotesize
\setlength{\tabcolsep}{2pt}
\renewcommand{\arraystretch}{0.96}
\begin{tabular*}{\columnwidth}{@{\extracolsep{\fill}}lc@{}}
\toprule
\textbf{Setting} & \textbf{Value / range} \\
\midrule
OPCs per run & $100$ \\
Type proportions (S1--S4) & $(.25,.25,.20,.20,.10)$ \\
Type proportions (S5) & $(.40,.10,.10,.30,.10)$ \\
Entry-cost multiplier / $\rho$ (S1--S5)
& $(1,1.8,1,.5,1)/(1,1,.5,1,1)$ \\
Outside-option coefficient / reporting scale $\zeta_W$ & $.20/100$ \\
Cycles / actor step / perturbation radius & $12/.30/.02$ \\
$H$ / $\tau_i$ / $\delta_\theta^{\max}$ & $5/.10/.15$ \\
$\delta_{\mathrm{LLM}}^{\max}$
& $\leq .05$ per component \\
Stability thresholds $(\epsilon_a,\epsilon_\theta,\epsilon_W)$
& $.03/.02/.15$ \\
$M_{\mathrm{stab}}$ / warm-up & $3/3$ cycles \\
Held-out skill cases & $240$ \\
Skill-library capacity $K_L$ & $20$ \\
Evolution cycles / quality / mutation / keep & $20/.60/.30/.75$ \\
LLM temperatures (cognition/dialogue/JSON) & $.30/.40/.10$ \\
\bottomrule
\end{tabular*}
\end{table}
\subsection{Baselines}
The baselines cover nonadaptive controls, online learning, objective-guided optimization, and consensus formation under matched populations, states, feasible bounds, and seeds. \textbf{No Policy} sets all policy intensities to zero and omits policy explanation. \textbf{Static Welfare} is our one-shot control: it searches feasible packages at the initial state and then fixes the best package, following the static-comparison role rather than the learning algorithm of AI Economist~\cite{zheng2022aieconomist}. \textbf{Contextual Bandit} adapts PolicySim's online-intervention setting and uses upper-confidence exploration over fixed policy-adjustment arms~\cite{li2010contextual,huang2026policysim}. \textbf{Carbon-PPO}, where ``Carbon'' denotes its carbon-market source domain, adapts the government planner of Wang et al.~\cite{wang2024carbon} to continuous regional-policy adjustments with clipped PPO updates. \textbf{Nash Welfare} replaces the aggregate objective with the geometric mean of type-level utilities~\cite{zhang2024noregret}. \textbf{Belief-Robust} adapts BIRDM by reweighting perturbed OPC-response scenarios and selecting the lowest-regret package~\cite{ciullo2023birdm}. \textbf{Bi-level Consensus} balances subgroup opinion adjustment and consensus cost under social influence~\cite{tu2024bilevel}. \textbf{CONSENSAGENT} uses repeated LLM discussions with sycophancy mitigation~\cite{pitre2025consensagent}. 

\subsection{Preliminary Study}
The preliminary study checks whether the simulator produces interpretable heterogeneity before evaluating PolicyMediator. With observed policies fixed, the simulated six-city ranking agrees with a public-data-derived entry-friendliness proxy (Spearman $\rho=.943$, $p=.0048$). Return decomposition shows that economic and industrial conditions establish most of the cross-city gradient, while policy support modifies rather than overturns structural suitability. Across OPC types within a city, initial cognitive assessments span up to $.30$, confirming that the profile model does not collapse heterogeneous firms to a common response.

Two further diagnostics qualify the ranking result. In a larger 500-OPC stress test, the four leading cities settle at similar entry shares of $22$--$24\%$ despite different initial attractiveness, because congestion redirects later entrants. Moreover, mean within-city type-level support balance ($.820$) exceeds cross-city spatial balance ($.774$); Hefei reaches high type-level balance ($.854$) while one professional-support type still has negative net utility ($-.09$). Thus balanced support can coexist with uniformly weak outcomes, and entry dispersion or balance alone should not be interpreted as welfare improvement. 

\section{Proofs}
\subsection{Existence of a Stationary Optimal Leader Policy}
\label{app:proof-stationary-leader}
\begin{proof}[Proof of Proposition~\ref{prop:stationary-leader-policy}]
For finite $A$ and $\tau_i>0$, Eq.~\eqref{eq:logit-response} has a positive denominator; payoff continuity thus makes each $\rho_i(\cdot\mid s,\mathbf u)$ and $\rho=\prod_i\rho_i$ continuous. Define
\begin{equation}
\begin{aligned}
\bar R^\rho(s,\mathbf u)
&=\sum_{\mathbf a}\rho(\mathbf a\mid s,\mathbf u)
W^G(s,\mathbf u,\mathbf a),\\
\bar P^\rho(\cdot\mid s,\mathbf u)
&=\sum_{\mathbf a}\rho(\mathbf a\mid s,\mathbf u)
P(\cdot\mid s,\mathbf u,\mathbf a).
\end{aligned}
\end{equation}
The assumptions make $\bar R^\rho$ bounded and upper semicontinuous and $\bar P^\rho$ weakly continuous. For feasible set $\mathcal U(s)$, the leader has Bellman operator
\begin{equation}
(\mathcal TV)(s)=
\max_{\mathbf u\in\mathcal U(s)}
\left\{\bar R^\rho(s,\mathbf u)
+\delta_G\int V(s')\,\bar P^\rho(ds'\mid s,\mathbf u)\right\}.
\end{equation}
Because $\delta_G<1$, $\mathcal T$ is a contraction on bounded value functions. Compact feasibility and the continuity conditions ensure that the maximum is attained and admits a measurable stationary selector $\pi_G^\star$. Its fixed point equals the response-induced return in Eq.~\eqref{eq:leader-induced-return}; hence $\pi_G^\star$ is optimal under the fixed logit response kernel.
\end{proof}

\FloatBarrier
\subsection{Zeroth-Order Gradient Estimation and Stationarity}
\label{app:proof-zeroth-order}
\begin{proof}[Proof of Proposition~\ref{prop:zeroth-order-stationarity}]
Condition on the saved state at cycle $k$. Unbiased rollouts, Gaussian symmetry, and Stein's identity give
\begin{equation}
\mathbb E_k[\widehat{\mathbf g}_k]
=\frac1h\mathbb E_{\boldsymbol\zeta}
[J_k(\boldsymbol\omega^k+h\boldsymbol\zeta)\boldsymbol\zeta]
=\nabla J_{k,h}(\boldsymbol\omega^k).
\end{equation}
If returns lie in a common interval of width $R$, centering them at its midpoint yields $L_h\leq R/(2h^2)$. Moreover, $\|\widehat{\mathbf g}_k\|^2\leq R^2\|\boldsymbol\zeta\|^2/(4h^2)$, so $\sigma_g^2\leq R^2D/(4h^2)$. The smoothness inequality and $\alpha L_h\leq1$ imply
\begin{equation}
\mathbb E_k[J_{k,h}(\boldsymbol\omega^{k+1})]
\geq J_{k,h}(\boldsymbol\omega^k)
+\frac{\alpha}{2}\|\nabla J_{k,h}(\boldsymbol\omega^k)\|^2
-\frac{L_h\alpha^2}{2}\sigma_g^2.
\end{equation}
Summing over $k$, telescoping the within-cycle gains, and bounding changes between consecutive saved-state objectives by $V_T$ proves Eq.~\eqref{eq:zeroth-order-stationarity}, with $\Delta=\sup_{k,\boldsymbol\omega}J_{k,h}(\boldsymbol\omega)-J_{0,h}(\boldsymbol\omega^0)$. If $\sum_k\alpha_k=\infty$, $\sum_k\alpha_k^2<\infty$, and total drift is summable, standard stochastic approximation gives $\|\nabla J_{k,h}(\boldsymbol\omega^k)\|\to0$ almost surely; the parameters themselves need not converge.
\end{proof}

\section{Ethical Considerations}
A central ethical concern of this work is preserving human agency in AI-assisted regional policymaking. PolicyMediator is designed to support policy deliberation and OPC location decisions, rather than replace governmental authority or entrepreneurial choice. Accordingly, policy adjustments should remain subject to institutional review, while OPCs retain autonomy over entry and relocation decisions. Moreover, optimizing aggregate government welfare does not necessarily ensure equitable benefits across regions or OPC types, making transparent assessment of distributional outcomes essential. From this perspective, our framework emphasizes informed coordination rather than enforced agreement. Finally, simulated outcomes provide analytical insights rather than direct evidence of real-world policy effectiveness, and practical deployment requires stakeholder evaluation and field validation.

\end{document}